\documentclass[11pt,twoside]{article}

\usepackage{fancyhdr}
\usepackage{amsmath, amssymb, amsthm}
\usepackage[initials,short-months,nobysame]{amsrefs}
\usepackage{tikz}

\theoremstyle{plain}
\newtheorem{theorem}{Theorem}[section]
\newtheorem{lemma}[theorem]{Lemma}
\newtheorem{corollary}[theorem]{Corollary}

\theoremstyle{definition}

\newtheorem*{algo}{Algebraic Reconstruction Algorithm}

\theoremstyle{remark}
\newtheorem{remark}{Remark}

\def\Hx{ \widehat{x} }
\def\Hy{ \widehat{y} }

\def\I{ \mathrm{i} }						
\def\E{ \mathrm{e} }						
\def\T{ \mathrm{T} }						
\def\RN{ \mathbb{R} }						
\def\CN{ \mathbb{C} }						
\def\uC{ \mathbb{T} }
\def\H{ \mathcal{H} }						
\def\NS{ \mathcal{N} }						
\def\RS{ \mathcal{R} }						
\DeclareMathOperator*{\cls}{\overline{span}}		
\DeclareMathOperator*{\rank}{rank}				
\DeclareMathOperator*{\trace}{Tr}				
\newcommand{\Op}[1]{\mathrm{#1}}

\begin{document}

\title{\bf\vspace{-39pt} Phase Retrieval from Low-Rate Samples\thanks{Preprint accepted for publication in \emph{Sampling Theory in Signal and Image Processing -- Special issue on SampTa~2013.}}
}

\author{
Volker Pohl \\
\small Lehrstuhl f{\"u}r Theoretische Informationstechnik\\
\small Technische Universit{\"a}t M{\"u}nchen\\
\small Arcisstrasse 21, 80290 M{\"u}nchen, Germany \\
\small volker.pohl@tum.de\\
\\
Fanny Yang \\
\small Department of Electrical Engineering and Computer Science\\
\small University of California-Berkeley\\ 
\small Berkeley, CA~94720, USA \\
\small fanny-yang@berkeley.edu\\
\\
Holger Boche\\
\small Lehrstuhl f{\"u}r Theoretische Informationstechnik\\
\small Technische Universit{\"a}t M{\"u}nchen\\
\small Arcisstrasse 21, 80290 M{\"u}nchen, Germany \\
\small boche@tum.de}

\date{July 18, 2014}

\maketitle
\thispagestyle{fancy}

\markboth{\footnotesize \rm \hfill V.~POHL, F.~YANG, H.~BOCHE \hfill}
{\footnotesize \rm \hfill PHASE RETRIEVAL FROM LOW-RATE SAMPLES \hfill}

\begin{abstract}
The paper considers the phase retrieval problem in $N$ dimensional complex vector spaces.
It provides two sets of deterministic measurement vectors which guarantee signal recovery for all signals, excluding only a specific subspace and a union of subspaces, respectively.
A stable analytic reconstruction procedure of low complexity is given.
Additionally it is proven that signal recovery from these measurements can be solved exactly via a semidefinite program.
A practical implementation with $4$ deterministic diffraction patterns is provided and some numerical experiments with noisy measurements complement the analytic approach.

\vspace{5mm}
\vspace{3mm}
\end{abstract}

\maketitle

\newpage
\section{Introduction}
\label{sec:Intro}

An object can be characterized by measuring its effect on the amplitude and phase of an electromagnetic wave. 
For very short wavelengths however, though it is easy to measure the intensity, phase information is usually hard to obtain.
Consequently the reconstruction of signals from intensity measurements alone (also known as \emph{phase retrieval}) is a very important problem in various fields of science and engineering, including 
X-ray crystallography \cites{Millane_90}, electron microscopy, astronomical imaging \cite{Fienup_93}, diffraction imaging, X-ray tomography but also in speech processing \cite{Balan_RecWithoutPhase_06}, radar \cite{Jaming_Radar10}, signal theory \cite{Oppenheim_Phase_80} or quantum tomography \cite{Heinosaari_QuantumTom_13}, to mention just a few.

The main problem in phase retrieval arises from the fact that generally the amplitude and phase of a signal are independent.
In order to overcome this problem, one may use, for example, prior knowledge about the signal to reconstruct the signal even without any phase measurements \cites{Ross_PhaseProblem78,Oppenheim_Phase_80}.
Another method to compensate for the missing phase information is the design of several different measurements of the same object under slightly different conditions.
This is a fairly popular and widely-used method in optics and implemented in very different ways, for example, by a distorted-object approach \cite{Xiao_DistortedObject05}, by aperture-plane modulations \cites{Zhang_ApatureMod07, Falldorf_SLM10}, or by recording several fractional Fourier transforms \cite{Jaming_Radar10} of the signal.
However, for a long time there was no systematic approach to design the different measurements such that exact signal recovery could be guaranteed.

Recently some remarkable contributions have been made in this discussion, stimulated mainly by \cite{Balan_RecWithoutPhase_06}.
In this work, it was shown that in an $N$-dimensional complex vector space, $4 N - 2$ intensity measurements are sufficient for phaseless signal recovery.
It has been conjectured in \cite{Bandeira_4NConj} that $4N-4$ generic intensity measurements are necessary and sufficient for signal recovery,
and it was shown in \cite{Conca_Algebraic13} that $4 N - 4$ measurements are indeed sufficient.
Explicit constructions of such measurement vectors were obtained in \cites{Bodmann_StablePR2014,Fickus_VeryFewMeasurements}, but
none of these papers provided a recovery algorithm which is also stable under noisy measurements.
Deterministic measurement vectors together with an analytic reconstruction algorithm was obtained in \cites{Balan_Painless_09}.
However, there the required number of measurements grows proportionally with $N^2$.
If one requires that only ``almost all'' vectors in $\CN^{N}$ can be recovered, then it is known that only $2 N$ measurements are necessary and sufficient \cites{Balan_RecWithoutPhase_06,Flammia_PureStates05,Finkelstein_QuantumCom04}. However, it was conjectured in \cite{Fickus_VeryFewMeasurements} that then phase retrieval is an NP-hard problem.

Ideas of sparse signal representation and convex optimization were used in \cites{CandesEldar_PhaseRetrieval,Candes_PhaseLift,Demanet_PhaselessLinMeas13} for phase retrieval based on $M$ random measurements, where $M$ is of the order $\mathcal{O}(N \log N)$.
It was shown that the corresponding recovery algorithm, now known as \emph{PhaseLift} \cite{Candes_PhaseLift}, provides stable recovery under noisy measurements.
This result was improved in \cite{CandesLi_FCM13} where it was shown that $\mathcal{O}(N)$ random measurements are sufficient to give a stable signal recovery via convex optimization.
Since SDP is computationally costly (e.g. SDPT3 \cite{Toh_SDP3_1999} needs $\mathcal{O}(N^{4.5}\log(\frac{1}{\epsilon})$ iterations), alternative algorithms for phase retrieval have been proposed which include PhaseCut \cite{Waldspurger_PR14} ($\mathcal{O}(N^{3}\sqrt{\log N}/\epsilon)$), alternating minimization \cites{Marchesini_AltProj14, Sanghavi_Alternating13} $\mathcal{O}(N^2)$, the fractional Fourier Transform \cite{Jaming_Fractional10} and polarization \cites{Alexeev_PhaseRetrieval13}.

All of the previously mentioned results address finite dimensional signal spaces.
For infinite dimensional spaces, it was shown in \cite{Thakur2011} that a real valued bandlimited signal can be recovered from magnitude samples taken at twice the Nyquist rate.
Complex valued $L^{2}$-signals with finite support are considered in \cite{Yang_SampTA13}. There recovery was guaranteed given specific amplitude measurements taken at four times the Nyquist rate. It provides a reconstruction algorithm which incorporates ideas from finite dimensional spaces and applies structured illuminations \cites{Xiao_DistortedObject05,Zhang_ApatureMod07, Falldorf_SLM10,CandesEldar_PhaseRetrieval}. This approach was extended to larger signal spaces in \cites{PYB_JFAA14, Pohl_ICASSP14, Yang_MsSc}.

In this work, we use the measurement design for infinite dimensional signal spaces from \cites{Yang_SampTA13,PYB_JFAA14,Yang_MsSc} to construct in Sec.~\ref{sec:Reconstruction} two sets of $4N-4$ measurement vectors for phase retrieval in the $N$-dimensional Euclidean space $\CN^{N}$.  
This set guarantees signal recovery in $\CN^{N}$, excluding a specific subspace or a specific union of subspaces.
For these measurement ensembles we provide an efficient algebraic recovery algorithm with a computational complexity of $\mathcal{O}(N)$, and Sec.~\ref{sec:Stability} will show that this algorithm is stable under noisy measurements.
In Sec.~\ref{sec:SDP} it is proven that signal recovery from these $4N-4$ deterministic measurements can also be obtained by a semidefinite program.
Finally, Sec.~\ref{sec:Simulations} provides some numerical experiments which illustrate the performance of the recovery algorithms in the presence of noise.
There we also compare our deterministic measurement vectors with ensembles of random measurement vectors as used in PhaseLift.

\section{Notations and Preliminary Results}
\label{sec:Notations}

We consider signals $x\in\CN^{N}$ in the $N$-dimensional complex Euclidean space which are denoted by  $x = (x[1], x[2], \dots, x[N])^{\T}$.
The inner product in $\CN^{N}$ is $\left\langle x,y \right\rangle_{\CN^{N}} = \sum^{N-1}_{n=1} x[n]\, \overline{y[n]} = y^{*}x$
where the bar denotes the complex conjugate, and $y^{*}$ is the conjugate transpose of $y$. The norm in $\CN^{N}$ is then $\|x\| = \sqrt{\left\langle x,x\right\rangle}$, and
the \emph{discrete Fourier transform (DFT)} of $x\in\CN^{N}$ is given by
\begin{eqnarray*}
	& \Hx[\omega] = (\mathcal{F} x)[\omega]
	= \sum^{N}_{t=1} x[t]\, \E^{-\I \frac{2\pi}{N} (\omega-1) (t-1) }\;,
	\quad \omega = 1,2,\dots,N\;.
\end{eqnarray*}
The unit circle in the complex plane will be denoted by $\uC = \{z \in \CN : |z| = 1\}$.

We write $\mathcal{H}_{N}$ for the Hilbert space of all $N\times N$ Hermitian matrices equipped with the \emph{Hilbert-Schmidt inner product} $\left\langle X,Y\right\rangle := \trace(Y^{*} X)$.
The induced (Frobenius) norm is denoted by $\|X\| = \sqrt{\left\langle X,X\right\rangle}$, whereas $\|X\|_{2}$ stands for the \emph{spectral norm} of $X$, for which $\|X\|_{2} \leq \|X\| \leq \sqrt{N} \|X\|_{2}$.
We will write $[X]_{m,n}$ for the entry in the $m$-th row and $n$-th column of the matrix $X$, and $I_{N}$ for the $N\times N$ identity matrix.

Let $v = \{v_{l}\}^{L}_{l=1}$ be a collection of vectors in $\CN^{N}$.
We consider the measurement mapping $\mathcal{A}_{v} : \CN^{N} \to \RN^{L}$ defined by
\begin{equation}
\label{equ:MeasureMap1}
	\mathcal{A}_{v} : x \mapsto \big\{ |\left\langle x , v_{l} \right\rangle|^{2} \big\}^{L}_{l=1}\;.
\end{equation}
Assume that $x\in\CN^{N}$ is arbitrary and assume that $b \in \RN^{L}$ is the vector which contains the known intensity measurements, i.e.
$b[l] = |\left\langle x, v_{l} \right\rangle|^{2}$ for $l = 1,\dots,L$.
Then the phase retrieval problem is to find $x \in \CN^{N}$ subject to $\mathcal{A}_{v}(x) = b$.
If $x$ is a solution to the phase retrieval problem then also $y = c x$ with $c \in \uC$ is also a solution.
For this reason, one considers the measurement process as a mapping $\mathcal{A}_{v} : \CN^{N}/\uC \to \RN^{L}$, where $\CN^{N}/\uC$ stands for the quotient space of $\CN^{N}$ modulo $\uC$.
So two vectors $x,y \in \CN^{N}$ are identified if there is a $c \in\uC$ such that $y=c\, x$.

The quadratic measurements \eqref{equ:MeasureMap1} of $x \in \CN^{N}$ can also be interpreted as linear measurements of the positive definite rank-one matrix $X = x x^{*}$.
Indeed, since
\begin{equation*}
	|\left\langle x,v_{l}  \right\rangle|^{2}
	= \trace(v^{*}_{l} x x^{*} v_{l})
	= \trace( V^{*}_{l} X )
	= \left\langle X , V_{l} \right\rangle
\end{equation*}
with $V_{l} = v_{l} v^{*}_{l}$, we can write \eqref{equ:MeasureMap1} as a linear mapping $\mathcal{A}_{V}:\mathcal{H}_{N} \to \RN^{L}$:
\begin{equation}
\label{equ:LinMeasureMap}
	\mathcal{A}_{V} : X \mapsto \left\{ \left\langle X , V^{*}_{l} \right\rangle \right\}^{L}_{l=1}\;.
\end{equation}
Therewith, the phase retrieval problem can be formulated as
\begin{equation}
\label{equ:PRProblem}
	\begin{array}{ll}
		\text{find} &  X\\
		\text{subject to} & \mathcal{A}_{V}(X) = b\;,\quad X \succeq 0\\
		& \rank(X) = 1\;.
	\end{array}
\end{equation}
The solution $X$ of \eqref{equ:PRProblem} can be factorized as $X = x x^{*}$ to recover the desired signal $x \in \CN^{N}$ up to a unitary constant.
This reformulation of the problem opens the way to solve the phase retrieval problem via a semidefinite program \cites{CandesRecht_MatrixCompl09,CandesEldar_PhaseRetrieval}.

Throughout the rest of the paper, we assume that $x \in \CN^{N}$ is the vector we seek to recover. With this vector we associate the linear subspace of $\mathcal{H}_{N}$
\begin{equation*}
	\mathcal{T} = \{ X = x y^{*} + y x^{*}\ :\  y\in\CN^{N} \}\;,
\end{equation*}
and $\mathcal{T}^{\bot}$ will denote the orthogonal complement of $\mathcal{T}$.
If we want to emphasize that $\mathcal{T}$ depends on the vector $x$, we will write $\mathcal{T}_{x}$.
The orthogonal projection of any $Y \in \mathcal{H}_{N}$ onto $\mathcal{T}$ and $\mathcal{T}^{\bot}$ is denoted by $Y_{\mathcal{T}} := \Op{P}_{\mathcal{T}}(Y)$ and $Y_{\mathcal{T}^{\bot}} := \Op{P}_{\mathcal{T}^{\bot}}(Y)$, respectively.
Finally, we notice that the adjoint $\mathcal{A}^{*}_{V} : \RN^{L} \to \mathcal{H}_{N}$ of the measurement mapping \eqref{equ:LinMeasureMap} is given by
\begin{equation*}
	\mathcal{A}^{*}_{V}(b)
	= \sum^{L}_{l=1} b[l]\, V_{l}
	= \sum^{L}_{l=1} b[l]\, v_{l} v^{*}_{l}\;.
\end{equation*}

The construction of our measurement vectors is based on ideas from \cite{Balan_Painless_09} where uniform $M/K$-tight frame where used for phase retrieval in $\CN^{K}$. The following property of such frames will be used frequently in this paper.

\begin{theorem}[\cite{Balan_Painless_09},\cite{Levenshtein_98}]
\label{thm:Balan}
Let $\{a_{1},\dots,a_{M}\}$ be a uniform $M/K$-tight frame in $\CN^{K}$ with $M=K^2$.
Then for every Hermitian rank-one matrix $Q_{x} = x x^{*} \in \mathcal{H}_{K}$ holds
\begin{equation}
\label{equ:BalanReconstruct}
	Q_{x} 
	= \frac{K+1}{K} \sum_{m=1}^{M} \left|\left\langle x,a_{m}\right\rangle\right|^{2}\, \left( a_{m}\, a^{*}_{m} - \frac{1}{K+1}\, I_{K} \right)\;.
\end{equation}
\end{theorem}

There exist explicit constructions for such frames for many different dimensions $K$ \cite{Zauner_Quantendesigns}.
We will only need the case $K=2$ for which a corresponding uniform $4/2$-tight frame $\{a_{1}, \dots,a_{4} \}$ is given by
\begin{equation}
\label{equ:alpha_2D}
	a_{1} =  \binom{\alpha}{\beta},\
	a_{2} =  \binom{\beta}{\alpha},\
	a_{3} =  \binom{\alpha}{-\beta},\
	a_{4} =  \binom{-\beta}{\alpha}
\end{equation}
with the constants
\begin{equation}
\label{equ:ab}
	\alpha = \sqrt{\tfrac{1}{2}\left( 1-\tfrac{1}{\sqrt{3}} \right)}
	\quad\text{and}\quad
	\beta = \E^{\I 5\pi/4}\sqrt{\tfrac{1}{2}\left( 1+\tfrac{1}{\sqrt{3}} \right)}.
\end{equation}

\section{Measurement Vectors and Reconstruction}
\label{sec:MeasureVectors}

In this section, we use ideas from \cites{Yang_MsSc,Yang_SampTA13,PYB_JFAA14} to construct two sets $\Phi$ and $\Psi$ with $L = 4 N - 4$ measurement vectors such that $\mathcal{A}_{\Phi}, \mathcal{A}_{\Psi} : \CN^{N}/\uC \to \RN^{L}$ are injective for generic $x \in \CN^{N}/\uC$.
More precisely the set $\Phi$ and $\Psi$ will yield a measurement mapping $\mathcal{A}_{\Phi}$ and $\mathcal{A}_{\Psi}$ which is 
injective on the dense subspace
\begin{equation}
\label{equ:Sets_S}
\begin{split}
	\mathcal{S}_{\Phi} &= \{ x \in \CN^{N}/\uC\ :\ x[n]\neq 0\ \text{for all}\ n=2,\dots,N-1\}\quad\text{and}\\
	\mathcal{S}_{\Psi} &= \{ x \in \CN^{N}/\uC\ :\ x[1]\neq 0\}\;,
\end{split}
\end{equation}
respectively.

\subsection{Injectivity and recovery algorithm}
\label{sec:Reconstruction}

Let $\{e_{n} = (0,\dots,0,1,0,\dots,0)^{\T}\}^{N}_{n=1}$ be the canonical orthonormal basis in $\CN^{N}$,
where the only non-zero entry of $e_{n}$ is at the $n$th position.
Therewith, we define two sets of $4N - 4$ measurement vectors in $\CN^{N}$:

\begin{equation}
\label{equ:MeasVect}
\begin{array}{rclcrcl}
	\phi_{1,n} & = & \phantom{-}\alpha\, e_{n} + \beta\, e_{n+1} & \quad , \quad &  \psi_{1,n} & = & \phantom{-}\alpha\, e_{1} + \beta\, e_{n+1}\\[0.5ex]
	\phi_{2,n} & = & \phantom{-}\beta\, e_{n} + \alpha\, e_{n+1} & \quad , \quad &  \psi_{2,n} & = & \phantom{-}\beta\, e_{1} + \alpha\, e_{n+1}\\[0.5ex]
	\phi_{3,n} & = & \phantom{-}\alpha\, e_{n} - \beta\, e_{n+1} & \quad , \quad &  \psi_{3,n} & = & \phantom{-}\alpha\, e_{1} - \beta\, e_{n+1}\\[0.5ex]
	\phi_{4,n} & = & -\beta\, e_{n} + \alpha\, e_{n+1}& \quad , \quad &  \psi_{4,n} & = & -\beta\, e_{1} + \alpha\, e_{n+1}
\end{array}
\end{equation}
for $n=1,2,\dots,N-1$ and where the constants $\alpha$ and $\beta$ are defined in \eqref{equ:ab}.

\begin{theorem}
\label{thm:FiniteDim}
Let $\Phi = \{ \phi_{m,n} \}^{m=1,\dots,4}_{n=1,\dots,N-1}$ and $\Psi = \{ \psi_{m,n} \}^{m=1,\dots,4}_{n=1,\dots,N-1}$ be the two sets of measurement vectors in $\CN^{N}$ defined in \eqref{equ:MeasVect}.
Then every $x \in \mathcal{S}_{\Phi}$ and every $y \in \mathcal{S}_{\Psi}$ can be recovered from the measurements
\begin{equation*}
	\mathcal{A}_{\Phi}(x) = \{\ |\left\langle x , \phi \right\rangle|^{2} : \phi\in\Phi\ \}
	\quad\text{and}\quad
	\mathcal{A}_{\Psi}(y) = \{\ |\left\langle y , \psi \right\rangle|^{2} : \psi\in\Psi\ \}\;,
\end{equation*}
respectively, up to a unitary constant. 
\end{theorem}

\begin{proof}
We begin with the statement for the set $\Phi$.
Let $x\in\CN^{N}$ with $x[n] \neq 0$ for all $n=2,\dots,N-1$ be arbitrary.
We have to show that we can recover $x$ from the $4 N - 4$ intensity measurements $\mathcal{A}_{\Phi}(x)$ up to a unitary factor $c \in \uC$.

For any fixed $n = 1,2,\dots,N-1$, we consider the four intensity measurements 
\begin{equation*}
	b_{m,n} = \left| \left\langle x , \phi_{m,n} \right\rangle \right|^{2}\;,
	\quad m=1,\dots,4\;.
\end{equation*}
Since all but $2$ entries of $\phi_{m,n}$ are equal to zero, these measurements can be written as
\begin{equation}
\label{equ:ProofMeasureAlp}
	b_{m,n} = \left| \left\langle x_{n} , a_{m}  \right\rangle_{\CN^{2}} \right|^{2}\;,
	\quad m=1,\dots,4\;.
\end{equation}
with the vectors $a_{m} \in \CN^{2}$ given in \eqref{equ:alpha_2D} and
\begin{equation*}
	x_{n} := (x[n] , x[n+1])^{\T}\;,
	\quad\text{for}\ n=1,2,\dots,N-1\;.
\end{equation*}
Recall that the set $\{a_{m}\}^{4}_{m=1}$ is a 2-uniform $2/4$-tight frame for $\CN^{2}$ (see \cite{Balan_RecWithoutPhase_06}) and notice that $x_{n}[2] = x_{n+1}[1]$.
For any $n=1,2,\dots,N-1$ define the $2\times 2$ rank-$1$ matrix $Q_{n} := x_{n}\,x^{*}_{n}$.
Then it follows from Theorem~\ref{thm:Balan} that
\begin{eqnarray}
\label{equ:Balan}
	& Q_{n} = \frac{3}{2} \sum^{4}_{m=1} b_{m,n} \left[ a_{m}\, a^{*}_{m} - \tfrac{1}{3} I_{2} \right]\,.
\end{eqnarray}
All values on the right hand side of \eqref{equ:Balan} are known so that for any fixed $n$, on can determine $Q_{n}$. 
Then $x_{n}$ can be determined up to a unitary factor with phase $\theta_n$ by factorizing $Q_{n}$ which yields $x_{n} \E^{\I\theta_{n}}$.

Now we start the recovery procedure with any $n \in \{2,3,\dots,N-1\}$.
As described above, we determine $x_{n} \E^{\I\theta_{n}}$ and set $\theta_{n} = \theta_{0}$ arbitrary.
Then we continue with $n+1$ and determine $x_{n+1} \E^{\I\theta_{n+1}}$ up to the unknown phase $\theta_{n+1}$.
However, since $x_{n}[2]= x_{n+1}[1]$ and $x_{n}[2] \neq 0$, we can determine the unknown phase $\theta_{n+1}$ from the already recovered vector $x_{n}$ by $\theta_{n+1} = \arg(x_{n}[2]) - \arg(x_{n+1}[1])$.
In the same way we continue with $n+2, n+3, \dots,N-1$, and similarly we can proceed in the other direction and continue with $n-1, n-2, \dots, 1$.
This way, it is possible to determine all vectors $\{x_{n}\}^{N-1}_{n=1}$ and consequently $\{x[n]\}^{N}_{n=1}$ up to the initial phase factor $\E^{\I\theta_{0}}$. 

The proof for $\Psi$ is almost identical. For each $n=1,2,\dots,N-1$ one considers the four intensity measurements $b_{m,n} = \left|\left\langle y,\psi_{m,n} \right\rangle\right|^{2}$ which can be written, similarly as in \eqref{equ:ProofMeasureAlp}, as $b_{m,n} = \left| \left\langle y_{n} , a_{m}  \right\rangle_{\CN^{2}} \right|^{2}$ but where $y_{n} := (y[1] , y[n+1])^{\T}$. As described above, we can recover all $y_{n} \E^{\I\theta_{n}}$, $n=1,2,\dots,N-1$ up to the unitary factors $\E^{\I\theta_{n}}$.
However, now all vectors $\{y_{n}\}^{N-1}_{n=1}$ contain $y[1]$ as their first entry, such that we can choose $\theta_{1}$ arbitrarily.
Then, as long as $y[1]\neq 0$, we can determine all other phases by $\theta_{n} = \arg(y_{1}[1]) - \arg(y_{n}[1])$.
\end{proof}

\begin{remark}
If the set of measurement vectors $\Psi$ is used then signal recovery will fail if a signal $x\in\CN^{N}$ is zero at its first position.
It is easily seen how the vectors $\Psi$ have to be changed to obtain a set $\Psi_{0}$ of measurement vectors
with the property that the corresponding measurement mapping $\mathcal{A}_{\Psi_{0}}$ is injective on $\mathcal{S}_{\Psi_{0}} = \{ x \in \CN^{N} / \uC : x[n_{0}]\neq 0\}$.
So the limitation on the set of signals which can not be reconstructed is very mild. In applications one only has to ensure that the signal does not vanish at one specific point.
\end{remark}

For clarity, we shortly summarize the reconstruction algorithm which was derived in the previous proof for the measurement ensembles $\Phi$ and $\Psi$.

\begin{algo}
Assume that the $4N - 4$ intensity measurements
$b_{m,n} = \left|\left\langle x, \phi_{m,n}\right\rangle\right|^{2}$ or $b_{m,n} = \left|\left\langle x, \psi_{m,n}\right\rangle\right|^{2}$
are given.
\begin{enumerate}
\item
Split the $N$-dimensional phase retrieval problem into $n=1,\dots,N-1$ two-dimensional problems
\begin{equation}
\label{equ:measurements1}
	b_{m,n} = \left|\left\langle x_n , a_m\right\rangle\right|^{2}\;,
	\quad m=1,\dots,4
\end{equation}
with either $x_{n} = (x[n], x[n+1])^{\T}$ or $x_{n} = (x[1], x[n+1])^{\T}$.
\item
For each $n = 1,2,\dots,N-1$ determine the matrix $Q_{n}$ in \eqref{equ:Balan} from the intensity measurements $b_{m,n}$.
\item Factorize these matrices as $Q_{n} = x_{n} x^{*}_{n}$. This yields $x_{n}\E^{\I\theta_n}$ with unknown phases $\theta_{n}$.
Practically, the factors can be calculated by determine the largest eigenvalue $\lambda_{\mathrm{max}}$ and the corresponding eigenvector $u$ of $Q_{n}$. Then
\begin{equation*}
	x_{n}\, \E^{\I\theta_{n}} = \sqrt{\lambda_{\mathrm{max}}}\, u\;.
\end{equation*}
\item Use the overlap between the vectors $x_{n}$ and $x_{n+1}$ or $x_{1}$ and $x_{n}$ to make the phases $\theta_{n}$ consistent over the whole vector $x = (x[1] , x[2], \dots, x[N])^{\T} \E^{\I \theta_{0}}$
with an overall unknown phase $\theta_{0}$.
\end{enumerate}
\end{algo}

Note that since the above recovery algorithm splits the $N$-dimensional phase retrieval into $N-1$ two-dimensional problems, its computational complexity scales linearly with the dimension $N$.

\subsection{Realization by modulations}
\label{sec:Realization}

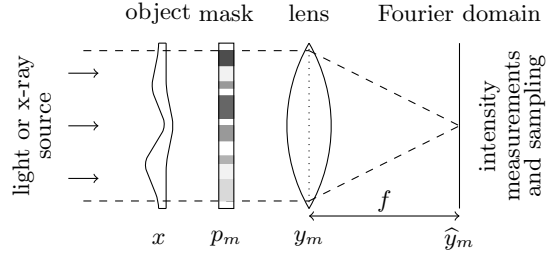
\begin{figure}[t]
\begin{center}
\begin{tikzpicture}
	\draw (-3.8,0) node[rotate = 90] {\footnotesize light or x-ray};
	\draw (-3.5,0) node[rotate = 90] {\footnotesize source};
	\draw[->] (-3.2,-0.7) -- (-2.8,-0.7);
	\draw[->] (-3.2,0) -- (-2.8,0);
	\draw[->] (-3.2,0.7) -- (-2.8,0.7);
	\draw[dashed] (-3,1) -- (0,1);
	\draw[dashed] (-3,-1) -- (0,-1);
	\draw[dashed] (0,1) -- (2,0);
	\draw[dashed] (0,-1) -- (2,0);
	\draw (-2,1.5) node {\footnotesize object};
	\draw[rounded corners = 1mm] (-2, -1.1) -- (-2,-1) -- (-2.2,-0.5) -- (-1.9,0) -- (-2.1,0.5) -- (-2,1) -- (-2,1.1);
	\draw[rounded corners = 1mm] (-1.9, -1.1) -- (-1.9,-1) -- (-1.9,-0.5) -- (-1.8,0) -- (-1.9,0.5) -- (-1.9,1) -- (-1.9,1.1);
	\draw (-1.9,-1.1) -- (-2,-1.1);
	\draw (-1.9,1.1) -- (-2,1.1);
	\draw  (-2,-1.5) node {\footnotesize $x$};
	\draw (-1.1,1.5) node {\footnotesize mask};
	\filldraw[fill = black!15!white, draw = black!15!white] (-1.2,-1) rectangle (-1,-0.7);
	\filldraw[fill = black!05!white, draw = black!05!white] (-1.2,-0.7) rectangle (-1,-0.4);
	\filldraw[fill = black!30!white, draw = black!30!white] (-1.2,-0.5) rectangle (-1,-0.4);
	\filldraw[fill = black!40!white, draw = black!40!white] (-1.2,-0.2) rectangle (-1,-0.0);
	\filldraw[fill = black!60!white, draw = black!60!white] (-1.2,0.1) rectangle (-1,0.4);
	\filldraw[fill = black!40!white, draw = black!40!white] (-1.2,0.5) rectangle (-1,0.6);
	\filldraw[fill = black!05!white, draw = black!05!white] (-1.2,0.6) rectangle (-1,0.8);
	\filldraw[fill = black!70!white, draw = black!70!white] (-1.2,0.8) rectangle (-1,1.0);
	\draw (-1.2,-1.1) rectangle (-1,1.1);
	\draw  (-1.1,-1.5) node {\footnotesize $p_{m}$};
	\draw  (0,1.5) node {\footnotesize lens};
	\draw (0,1.1) arc (150:210:2.2);
	\draw (0,-1.1) arc (-30:30:2.2);
	\draw[dotted] (0,-1.1) -- (0,1.1);
	\draw  (0,-1.5) node {\footnotesize $y_{m}$};
	\draw  (2,1.5) node {\footnotesize Fourier domain};
	\draw (2,1.1) -- (2,-1.1);
	\draw  (2,-1.5) node {\footnotesize $\Hy_{m}$};
	\draw[<->] (0,-1.2) -- (2,-1.2);
	\draw (1,-1) node {\footnotesize $f$};
	\draw (2.4,0) node[rotate = 90] {\footnotesize intensity};
	\draw (2.7,0) node[rotate = 90] {\footnotesize measurements};
	\draw (3.0,0) node[rotate = 90] {\footnotesize and sampling};
\end{tikzpicture}
\end{center}
\caption{Typical setup in several imaging applications using masks for structured illuminations.}
\label{fig:Optics}
\end{figure}

In applications, one often has a measurement setup as in Fig.~\ref{fig:Optics} (see, e.g., \cites{BandChen_II14,Candes_CDP13,CandesEldar_PhaseRetrieval,Falldorf_SLM10,Gross_2014,Xiao_DistortedObject05,Zhang_ApatureMod07}). 
There a certain object is illuminated by a light or x-ray source. This produces a diffraction pattern $x[t]$, where $t$ stands for the spatial coordinate. Then several masks are insert behind the object. These masks have a certain transmittance function $p_{m}[t]$ and modulate the diffraction pattern as $y_{m}[t] = x[t] p_{m}[t]$. The subsequent lens transforms $y_{m}$ into the Fourier domain and there the squared modulus of $|\widehat{y}_{m}[\omega]|^{2}$ is measured.

Here we show that the measurement vectors \eqref{equ:MeasVect} may be implemented using a setup as in Fig.~\ref{fig:Optics} by choosing the masks $p_{m} \in \CN^{N}$ appropriately.
In particular, we show that there exist $4$ masks $p_{m}[t]$ such that
\begin{equation*}
	\left| \mathcal{F}\left( x[t]\, p_{m}[t] \right)[n] \right|^{2}
	= b_{m,n}
	= \left|\left\langle \Hx , \phi_{m,n} \right\rangle\right|^{2}\,,
	\quad m=1,\dots,4\;; n=1,\dots,N\;.
\end{equation*}
where $\phi_{m,n}$ are the measurement vectors given in \eqref{equ:MeasVect}, and the same will be shown for the vectors $\psi_{m,n}$.
Then $\Hx = \mathcal{F} x$ can be reconstructed with the previous phase retrieval procedure and $x$ is obtained by the inverse DFT from $\Hx$.

\paragraph{Measurement vectors $\Phi$}
We choose the four masks $p_{m} \in \CN^{N}$ as follows
\begin{equation*}
	p_{m}[t] = \overline{a_{m}[1]} + \overline{a_{m}[2]}\, \E^{-\I \frac{2\pi}{N} (t-1)}\;,
	\quad t = 1,2,\dots,N \quad\text{and}\quad m=1,\dots,4\;.
\end{equation*}
where the coefficient vectors $a_{m} = (a_{m}[1] , a_{m}[2])^{\T}$ are chosen as in \eqref{equ:alpha_2D}.
Setting $y[t] = x[t]\, p_{m}[t]$ and taking the DFT, one obtains for all $n=1,2,\dots,N-1$
\begin{equation*}
	\Hy_{m}[n]
	= (\mathcal{F} y)[\omega]
	= \overline{ a_{m}[1] }\, \Hx[n] + \overline{ a_{m}[2] }\, \Hx[n+1]
	= \left\langle \Hx , \phi_{m,n} \right\rangle_{\CN^{N}}\;.
\end{equation*}
These are the frequency measurement using the $m$-th mask in the sampling system in Fig.~\ref{fig:Optics}.
Therewith, the intensity measurements become
\begin{equation}
\label{equ:IntensMeasFD}
	b_{m,n}
	= \left| \Hy_{m}[n] \right|^{2}
	= \left| \left\langle \Hx , \phi_{m,n} \right\rangle_{\CN^{N}} \right|^{2}\,,
	\quad
	\begin{array}{l}
		m=1,\dots,4\\
		n=1,2,\dots,N-1\;,
	\end{array}
\end{equation}
where $\phi_{m,n}$ are exactly the same measurement vectors as given in \eqref{equ:MeasVect}.
Now one can recover $\Hx$ from the measurements \eqref{equ:IntensMeasFD} either using the algorithm presented in Sec.~\ref{sec:Reconstruction} or by an SDP as discussed in Section~\ref{sec:SDP} below.

\paragraph{Measurement vectors $\Psi$}
To implement the measurement vectors $\Psi$, we choose the masks $p_{m} \in \CN^{N}$ as
\begin{equation*}
	p_{m}[t] = \overline{a_{m}[1]}\, \delta[t] + \overline{a_{m}[2]}\;,
	\quad t = 1,2,\dots,N \quad\text{and}\quad m=1,\dots,4\;,
\end{equation*}
where $\delta[t]$ is the delta function defined by $\delta[1] = 1$ and $\delta[t] = 0$ for $t = 2,3,\dots,N$.
Multiplying $x$ with $p_{m}$ and taking the DFT, one gets
\begin{equation*}
	\widehat{y}_{m}[n]
	= \overline{a_{m}[1]}\, x[0] + \overline{a_{m}[2]}\, \widehat{x}[n]
	= \left\langle \widetilde{x} , \psi_{m,n} \right\rangle_{\CN^{N+1}}\;,
	\quad n=1,2,\dots,N\;,
\end{equation*}
for the $n$-th frequency measurement using the $m$-th mask and
where $\widetilde{x} := (x[0],\Hx^{\T})^{\T} \in \CN^{N+1}$ and where $\psi_{m,n}$ are measurement vectors as given  in \eqref{equ:MeasVect} but for $\CN^{N+1}$.
Then the intensity measurements are 
\begin{equation}
\label{equ:IntensMeasFD2}
	b_{m,n}
	= \left| \Hy_{m}[n] \right|^{2}
	= \left| \left\langle \widetilde{x} , \psi_{m,n} \right\rangle_{\CN^{N+1}} \right|^{2}\,,
	\quad
	\begin{array}{l}
		m=1,\dots,4\\
		n=1,2,\dots,N\;.
	\end{array}
\end{equation}
Compared with the previous setup for $\Phi$, we need $4 N$ measurements since additionally $x[0]$ is determined.
However, the point $x[0]$ is only needed as a ``Punctum Archimedis'' to derive the unknown phase of each block from this point and to match the unknown phases between the different blocks. 
This way, we avoid the phase propagation which is necessary if the measurement vectors $\Phi$ are used and which yields a poor error performance, as it will be discussed in Sec.~\ref{sec:Stability} and \ref{sec:Simulations}.

\section{Stability Analysis}
\label{sec:Stability}

This section analyzes the stability behavior of the recovery algorithm of Sec.~\ref{sec:Reconstruction}.
To this end, we suppose that the measurements are disturbed by additive noise:
\begin{equation}
\label{equ:NoisyMeas}
	\widetilde{b}_{m,n} = |\left\langle x,\phi_{m,n}\right\rangle|^2 + \nu_{m,n}\;,
	\quad
	\begin{array}{l}
	m=1,\dots,4\\
	n=1,\dots,N-1\;.
	\end{array}
\end{equation}
It is assumed that the noise terms $\nu_{m,n}$ are real valued and independent, identical distributed (i.i.d) random variables.
All noise components $\nu_{m,n}$ are collected in the vector $\nu \in \RN^{4 (N-1)}$.
Now signal reconstruction will be based on the disturbed values $\widetilde{b}_{m,n}$ which will give an erroneous reconstructed vector $\widetilde{x}$.
Based on these assumptions, we will derive an upper bound on the expected squared error $E[ \|x - \widetilde{x}\|^{2} ]$ as a function of the average squared norm $E[ \|\nu\|^{2} ]$ of the noise.

The algorithm of Sec.~\ref{sec:Reconstruction} splits the $N$-dimensional phase retrieval into $N-1$ two-dimensional problems. Therefore, the first subsection analyzes the two-dimensional phase retrieval
and the phase propagation before error bounds for the $N$-dimensional problem are derived in Subsection~\ref{sec:ErrorBounds}.

\subsection{Some preliminary results}

Under Point~$2$ of the reconstruction algorithm, one determines the matrix $Q_{n}$ from the noise intensity measurements $\{\widetilde{b}_{m,n}\}^{M}_{m=1}$ for each $n = 1,\dots,N-1$. 
This gives an erroneous matrix $\widetilde{Q}_{n} = Q_{n} + \Delta Q$ with
\begin{equation}
\label{equ:ErrorMatrix}
	\Delta Q
	= \frac{K+1}{K} \sum_{m=1}^{M} \nu_{m}\, \left( A_{m} - \frac{1}{K+1}\, I_{K} \right)
	= \frac{K+1}{K} \sum_{m=1}^{M} \nu_{m}\, \widetilde{A}_{m}
\end{equation}
where we defined $A_{m} := a_{m}\, a^{*}_{m}$ and $\widetilde{A}_{m} := A_{m} - \tfrac{1}{K+1}\, I_{K}$, and where we omitted the subscript $n$ at the noise terms.
In the following we write $\nu = (\nu_{1}, \dots,\nu_{m})^{\T} \in \RN^{M}$ for the vector containing all noise terms in step $n$.
The next lemma derives bounds on the norm of the matrix $\Delta Q$. The lemma is formulated for matrices of arbitrary size $K$. Later we only need the case $K=2$.

\begin{lemma}
\label{lem:ErrorMatrix}
Let $\{a_{m}\}^{M}_{m=1}$ be a $2$-uniform $M/K$-tight frame with $M = K^{2}$ vectors, and let $\Delta Q(\nu) = \Delta Q(\nu)$ be the matrix \eqref{equ:ErrorMatrix}, then
\begin{equation}
\label{equ:LemmaErrorMatrix}
	\|\nu\| \leq \|\Delta Q\| \leq \sqrt{1 + \tfrac{1}{K}}\, \|\nu\|\;.
\end{equation}
\end{lemma}

\begin{proof}
By the definition of the Hilbert-Schmidt norm and \eqref{equ:ErrorMatrix}, we have
\begin{equation*}
	\|\Delta Q\|^{2}
	= \left\langle \Delta Q , \Delta Q \right\rangle
	= \left(\frac{K+1}{K}\right)^{2} \sum^{K^{2}}_{m=1}\sum^{K^{2}}_{n=1} \nu_{m}\nu_{n}\, \big\langle \widetilde{A}_{m} , \widetilde{A}_{n} \big\rangle\;.
\end{equation*}
Moreover, the inner products on the right hand side are given by
\begin{align*}
	\big\langle \widetilde{A}_{m} , \widetilde{A}_{n} \big\rangle
	&= \big\langle A_{m} - \tfrac{1}{K+1}\, I_{K}, A_{n} - \tfrac{1}{K+1}\, I_{K} \big\rangle\\
	&= \trace(A_{m}A_{n}) - \tfrac{1}{K+1}\left[ \trace(A_{m}) + \trace(A_{n}) \right] + \tfrac{1}{(K+1)^{2}}\trace(I_{K})\\
	&= \left|\left\langle a_{m} , a_{n}\right\rangle\right|^{2} - \tfrac{1}{K+1}\left[ \|a_{m}\|^{2}  + \|a_{n}\|^{2} \right] + \tfrac{K}{(K+1)^{2}}\;.
\end{align*}
By the assumption on the set $\{a_{m}\ :\ m=1,\dots,K^{2}\}$, we have (see, e.g., \cite{Balan_Painless_09})
\begin{equation*}
	\left|\left\langle a_{m} , a_{n}\right\rangle\right|^{2}
	= \left\{\begin{array}{cll}
	1 & \text{if} & m=n\\
	\frac{1}{K+1} & \text{if} & m\neq n
	\end{array}\right.
\end{equation*}
such that
\begin{equation*}
	\big\langle \widetilde{A}_{m} , \widetilde{A}_{n} \big\rangle
	= \left\{\begin{array}{cll}
	\frac{K^{2} + K - 1}{(K+1)^{2}} & \text{if} & m=n\\[1ex]
	-\frac{1}{(K+1)^{2}} & \text{if} & m\neq n
	\end{array}\right.
\end{equation*}
and therefore
\begin{equation*}
	\|\Delta Q\|^{2}
	= \frac{1}{K^{2}} \left( \left(K^{2} + K - 1\right) \sum^{K^2}_{m=1} |\nu_{m}|^{2} - \sum^{K^{2}}_{\substack{ m,n=1 \\ m\neq n}} \nu_{m}\, \nu_{n}\right)\;.
\end{equation*}
The last equation can also be written as $\|\Delta Q\|^{2} = \frac{1}{K^{2}}\, \nu^{T} B \nu$, where $B$ is an $M \times M$ matrix with identical diagonal entries $[B]_{n,n} = K^{2} + K - 1$ for all $n = 1,\dots,K^{2}$ and with all off-diagonal entries equal to $-1$. Then we apply the Rayleigh-Ritz theorem (see, e.g., \cite{HornJohnson}) to obtain
\begin{equation}
\label{equ:DQ2}
	\frac{1}{K^{2}}\, \lambda_{\mathrm{min}}(B)\, \|\nu\|^{2}\
	\leq\ \|\Delta Q\|^{2}
	= \frac{1}{K^{2}}\, \nu^{T} B \nu\
	\leq\ \frac{1}{K^{2}}\, \lambda_{\mathrm{max}}(B)\, \|\nu\|^{2}
\end{equation}
where $\lambda_{\mathrm{min}}(B)$ and $\lambda_{\mathrm{max}}(B)$ stands for the smallest and the largest eigenvalue of $B$, respectively.
Since $B$ is a circulant matrix, its eigenvalues are given as the DFT of its first row \cite{Gray_ToeplitzMatrices} such that $\lambda_{\mathrm{min}}(B) = K^{2}$ and $\lambda_{\mathrm{max}}(B) = K(K+1)$ is obtained.
Inserting this in \eqref{equ:DQ2} one obtains \eqref{equ:LemmaErrorMatrix}.
\end{proof}

Under Point~$3$ of the reconstruction algorithm the matrix $\widetilde{Q}_{n} = Q_{n} + \Delta Q$ is factorized to obtain an estimate
$\widetilde{x}_{n} = \widetilde{\lambda}^{1/2}_{\mathrm{max}}\, \widetilde{u}_{n}$ of the vector $x_{n} \in \CN^{K}$.
Therein $\widetilde{\lambda}_{\mathrm{max}}$ is the largest eigenvalue of $\widetilde{Q}_{n}$ and $\widetilde{u}_{n}$ is the corresponding eigenvector.
Since $\widetilde{Q}_{n} \neq Q_{n} = x_{n} x^{*}_{n}$, the eigenvalue $\widetilde{\lambda}_{\mathrm{max}}$ and the eigenvector $\widetilde{u}_{n}$ will not be equal to the true eigenvalue $\lambda_{\mathrm{max}} = \|x_{n}\|^{2}$ and the eigenvector $u_{n} = x_{n}/\|x_{n}\|$, respectively, of $Q_{n}$.
The next lemma gives an upper bound on the error $\|x_{n} - \widetilde{x}_{n}\|_{2}$ in terms of the norm of the error matrix $\Delta Q$. Again, the lemma is formulated for vectors in $\CN^{K}$ although later only the case $K=2$ is used. For simplicity of notation, the subscribe $n$ will be omitted.

\begin{lemma}
\label{lem:EstErrEigenv}
For some $x \in \CN^{K}$ let $Q_{x} = x x^{*}$, and $\widetilde{Q}_{x} = Q_{x} + \Delta Q$ with $\Delta Q \in \H_{K}$.
Let $\widetilde{\lambda}_{\mathrm{max}}$ be the largest eigenvalue of $\widetilde{Q}_{x}$ and $\widetilde{u}$ the corresponding eigenvector. 
If $\widetilde{x} = \widetilde{\lambda}^{1/2}_{\mathrm{max}}\, \widetilde{u}$, then
\begin{equation}
\label{equ:ErrorBound2}
	\| x\, \E^{\I\theta} - \widetilde{x}\|_{2}
	\leq \left\{ \begin{array}{lll}
	2\, \frac{\| \Delta Q \|_{2}}{\|x\|} \leq 2\, \frac{\| \Delta Q \|}{\| x\|} &\ \text{if}\ & \|x\| \geq 3\, \|\Delta Q\|_{2}\\[1.5ex]
	\sqrt{7 \|\Delta Q\|_{2}} \leq \sqrt{7 \|\Delta Q\|} &\ \text{if}\ & \| x\| < 3\, \|\Delta Q\|_{2}
	\end{array}\right.
\end{equation}
for some $\theta \in [-\pi,\pi)$.
\end{lemma}

\begin{proof}
The proof follows basically the ideas in \cite{Candes_PhaseLift}.
Since both matrices $Q_{x}$ and $\Delta Q$ are self-adjoint, Weyl's inequality (see, e.g., \cite{HornJohnson}*{Chap.~4.3}) gives
\begin{equation}
\label{equ:errorEW}
	\big| \lambda_{\mathrm{max}} - \widetilde{\lambda}_{\mathrm{max}} \big| \leq \|\Delta Q\|_{2} =: \epsilon\,.
\end{equation}
Moreover, the $\sin$-$\theta$-Theorem \cite{Davis_RotEV_70} provides an upper bound on the angle $\theta$ between the eigenvectors $u$ and $\widetilde{u}$. It states that
\begin{equation}
\label{equ:errorSinTheta}
	\left| \sin(\theta) \right|
	\leq \frac{\|\Delta Q\|_{2}}{|\widetilde{\lambda}_{\mathrm{max}} |}
	\leq \frac{\|\Delta Q\|_{2}}{\left|\, \lambda_{\max} - \|\Delta Q\|_{2}\, \right|}
	= \frac{\epsilon}{\left|\, \|x\|^{2} - \epsilon\, \right|}
\end{equation}
where the second inequality follows from \eqref{equ:errorEW}.
Now one can decomposes the vector $\widetilde{u}$ as $\widetilde{u} = \cos(\theta) u + \sin(\theta) u^{\bot}$
into a component parallel to $u$ and a component perpendicular to $u$.
Then $x - \widetilde{x} = [ \|x\|  - \widetilde{\lambda}^{1/2}_{\mathrm{max}}\, \cos(\theta) ]\, u - \widetilde{\lambda}^{1/2}_{\mathrm{max}}\, \sin(\theta)\, u^{\bot}$
and Pythagoras' formula gives
\begin{equation}
\label{equ:errorNorm}
	\| x - \widetilde{x} \|^{2}
	= \left( \|x\|  - \sqrt{\widetilde{\lambda}_{\mathrm{max}}}\, \cos(\theta)\right)^{2} + \widetilde{\lambda}_{\mathrm{max}}\, \sin^{2}(\theta)\;.
\end{equation}
For the second term on the right hand side, we easily get from \eqref{equ:errorEW} and \eqref{equ:errorSinTheta}
\begin{equation*}
	 \widetilde{\lambda}_{\mathrm{max}}\, \sin^{2}(\theta)
	\leq (\|x\|^{2} + \epsilon)\, \frac{\epsilon^{2}}{ (\|x\|^{2} - \epsilon)^{2}}
	= \frac{1 + \gamma}{(1 - \gamma)^{2}}\, \frac{\epsilon^{2}}{\|x\|^{2}}
\end{equation*}
with $\gamma := \epsilon/\|x\|^{2}$.
To get an upper bound on the first term on the right hand side of \eqref{equ:errorNorm}, we notice that
\begin{align*}
	\sqrt{\widetilde{\lambda}_{\mathrm{max}}}\, \cos(\theta)
	&= \sqrt{\widetilde{\lambda}_{\mathrm{max}}}\, \sqrt{ 1 - \sin^{2}(\theta)}
	\geq \sqrt{ (\|x\|^{2} - \epsilon) \left(1 - \frac{\epsilon^{2}}{(\|x\|^{2} - \epsilon)^{2}}\right) }\\
	&= \|x\| \sqrt{\frac{1-2\gamma}{1-\gamma}}
	\geq \|x\| (1-\gamma)
\end{align*}
where the last inequality holds for all $0 \leq \gamma \leq (3 - \sqrt{5})/2$.
Therewith the estimation error \eqref{equ:errorNorm} can be upper bounded by
\begin{equation*}
	\|x - \widetilde{x}\|^{2}
	\leq \left( \|x\|^{2} - \|x\|^{2} [1 - \gamma]\right)^{2} + \frac{1 + \gamma}{(1 - \gamma)^{2}}\, \frac{\epsilon^{2}}{\|x\|^{2}}
	= \left( 1 +  \frac{1 + \gamma}{(1 - \gamma)^{2}} \right)\frac{\epsilon^{2}}{\|x\|^{2}}\;.
\end{equation*}
For $\gamma \leq 1/3$ one obtains in particular $\|x - \widetilde{x}\|^{2} \leq 4 \epsilon^{2}/\|x\|^{2}$.
If $\gamma > 1/3$ we use the trivial estimate $|\sin(\theta)| \leq 1$.
Then \eqref{equ:errorNorm} and \eqref{equ:errorEW} give
$\|x - \widetilde{x}\|^{2} \leq 2 \|x\|^{2} + \epsilon \leq 7\, \epsilon$.
Taking the square root, one obtains \eqref{equ:ErrorBound2}.
\end{proof}

\begin{remark}
The first line of \eqref{equ:ErrorBound2} describes the high signal-to-noise (SNR) range, whereas the second line gives an error estimate for low SNR.
According to \eqref{equ:ErrorBound2}, the upper bound in the low SNR regime is only determined by the size of the disturbance $\|\Delta Q\|$ but it is independent of the signal.
\end{remark}

Now we consider the error propagation due to phase propagation under Point~$4$ of the reconstruction algorithm.
Point~$3$ of the algorithm determines an estimate $\widetilde{x}_{n}$ of the vector $x_{n}$, and Lemma~\ref{lem:EstErrEigenv} gives an error bound for the corresponding estimation error $\|x_{n} \E^{\I\theta_{n}} - \widetilde{x}_{n}\|^{2}$.
This bound holds for the optimal phase
$\theta_{n} = \arg \min_{\theta \in [-\pi,\pi]} \|x_{n}\E^{\I\theta} - \widetilde{x}_{n}\| $
of the vector $x_{n}$.
In the following we subsume $x_{n}$ and the optimal phase factor $\E^{\I\theta_{n}}$ and simply write $x_{n}$ for $x_{n}\E^{\I\theta_{n}}$.
The phase information of $x_{n}$ is derived from the previous vector $\widetilde{x}_{n-1}$ or from the vector $\widetilde{x}_{1}$ in Point~$4$ of our algorithm. Since the estimation of these vectors is already erroneous, also the phase information will be erroneous.
As a consequence, the overall estimation error in step $n$ of our algorithm is
\begin{equation}
\label{equ:error_xn}
	\|x_{n}\, \E^{\I\Delta\theta_{n}} - \widetilde{x}_{n}\|^{2}
	\leq 2\, \|x_{n} - \widetilde{x}_{n}\|^{2} + 2\, \|x_{n} - x_{n}\, \E^{\I\Delta\theta_{n}}\|^{2}
\end{equation}
where $\Delta\theta_{n}$ is the described phase error and the above inequality is obtained from the parallelogram law.
Lemma~\ref{lem:EstErrEigenv} provides an upper bound for the first term on the right hand side. To derive a bound for the second term, we notice that $ \|x_{n} - x_{n}\, \E^{\I\Delta\theta_{n}}\|^{2} = 2 \|x_{n}\|^2\, (1 - \cos\Delta\theta_n)$. Now, the phase error $\Delta\theta_{n}$ is due to an estimation error $\|x_{n-1}\E^{\I\Delta\theta_{n-1}} - \widetilde{x}_{n-1}\|^{2}$ in the previous step.
In the worst case this estimation error could entirely be written as a phase error as
\begin{equation*}
	\|x_{n-1}\E^{\I\Delta\theta_{n-1}} - \widetilde{x}_{n-1}\|^{2}
	= \|x_{n-1} - x_{n-1}\E^{\I\Delta\theta}\|^{2}
	= 2\, \|x_{n-1}\|^{2} (1 - \cos\Delta\theta)\;.
\end{equation*}	
Then our algorithm would give $\Delta\theta_{n} = \Delta\theta$ and so the whole estimation error in step $n-1$ is propagated to step $n$.
So for the worst case, we have the estimate
\begin{equation*}
	 \|x_{n} - x_{n} \E^{\I\Delta\theta_{n}}\|^{2}
	\leq \frac{\|x_{n}\|^{2}}{\|x_{n-1}\|^{2}}\, \|x_{n-1}\E^{\I\Delta\theta_{n-1}} - \widetilde{x}_{n-1}\|^{2}
	\leq \gamma\, \|x_{n-1} \E^{\I\Delta\theta_{n-1}} - \widetilde{x}_{n-1}\|^{2}
\end{equation*}
where $\gamma$ is a certain upper bound on the ratio $\|x_{n}\|^{2} / \|x_{n-1}\|^{2}$ which describes is a sense the allowed variability in our signal amplitude.
Overall \eqref{equ:error_xn} becomes
\begin{equation}
\label{equ:errorStepN}
	\|x_{n}\, \E^{\I\Delta\theta_{n}} - \widetilde{x}_{n}\|^{2}
	\leq 2\, \left( \|x_{n} - \widetilde{x}_{n}\|^{2} + \gamma\, \|x_{n-1}\E^{\I\Delta\theta_{n-1}} - \widetilde{x}_{n-1}\|^{2} \right)\;.
\end{equation}
So due to the phase propagation in our reconstruction algorithm, a portion of the estimation error in the previous step is propagated to the actual step. 
Generally $\gamma$ may be viewed as the ratio of the error energy which is propagated from step to step and an upper bound for $\gamma$ is given by
\begin{equation*}
	\gamma \leq \frac{\|x\|^{2}_{\infty}}{\mu(x)^{2}}
	\qquad\text{with}\qquad
	\mu(x) = \left\{\begin{array}{ll}
		\min_{ n\in\{2,3,\dots,N-1\} } |x[n]| & \text{for}\ \Phi\\[1ex]
		|x[1]| & \text{for}\ \Psi
	\end{array}\right.\;.
\end{equation*}
Note that $\mu(x) > 0$ since our signals $x$ belong to the subspaces \eqref{equ:Sets_S}.

\subsection{Bounds on the reconstruction error}
\label{sec:ErrorBounds}

After these preparations, we are ready to prove error bounds for the recovery algorithm of Sec.~\ref{sec:Reconstruction}. 
These bounds depend on the actual ensemble of measurement vectors.
We start with the result for $\Psi$.

\begin{theorem}
\label{thm:ERRBound_Psi}
Let $\Psi = \{\psi_{m,n}\}$ be the measurement ensemble defined in \eqref{equ:MeasVect}. For any $\mu  > 0$ let
$x \in \mathcal{S}_{\Psi}(\mu) := \{ x \in \mathcal{S}_{\Psi} : |x[1]| \geq \mu\}$
be arbitrary and let $\widetilde{x}$ be the reconstructed vector using the algorithm given in Sec.~\ref{sec:Reconstruction}.\\
If the measurement noise $\nu_{m,n}$ in \eqref{equ:NoisyMeas} are i.i.d. random variables
then there are constants $C_{1} = C_{1}(\mu)$ and $C_{2} = C_{2}(\mu)$ such that
\begin{equation}
\label{equ:ErrBound2}
	E\left[ \|x - \widetilde{x}\|^{2} \right]
	\leq \left\{\begin{array}{lll}
	C_{1}\, E[ \|\nu\|^{2} ] & \quad & \text{for small}\ E[ \|\nu\| ]\\[1.5ex]
	C_{2}\, \sqrt{N-1}\, E[\|\nu\|] & \quad & \text{for large}\ E[\|\nu\|]
	\end{array}\right.
\end{equation}
and with $C_{1} \leq 12\,(1+\gamma)/\mu^{2}$, $C_{2} \leq 14 \sqrt{3/2}\,(1+\gamma)$ and with $\gamma \leq \|x\|^{2}_{\infty}/\mu^{2}$.
\end{theorem}

\begin{remark}
The constants $C_{1}$ and $C_{2}$ depend in particular on the amplitude of the first signal component $x[1]$.
To get small constants (i.e. a low mean squared error), $x[1]$ should have the largest amplitude among all entries $x[n]$ of $x$.
This might be achieved by an adequate measurement setup.
\end{remark}

\begin{proof}
Since $x_{n} = (x[1],x[n])^{\T}$, the overall estimation error is given by
\begin{multline}
\label{equ:proofErrBound1}
	\|x - \widetilde{x}\|^{2}
	= \sum^{N-1}_{n=1} \|x_{n}\E^{\I\Delta\theta_{n}} - \widetilde{x}_{n}\|^{2} - (N-2) \left|x[1] - \widetilde{x}[1]\right|^{2}\\
	\leq \sum^{N-1}_{n=1} \|x_{n}\E^{\I\Delta\theta_{n}} - \widetilde{x}_{n}\|^{2}
	\leq 2 \sum^{N-1}_{n=1} \|x_{n} - \widetilde{x}_{n}\|^{2} + 2\, \gamma \sum^{N-1}_{n=2} \|x_{1} - \widetilde{x}_{1}\|^{2}
\end{multline}
using \eqref{equ:errorStepN} for the last line and that $\Delta\theta_{1} = 0$ since the initial phase is unknown.
Combining Lemma~\ref{lem:ErrorMatrix} and \ref{lem:EstErrEigenv}, we have
\begin{equation*}
	\|x_{n} - \widetilde{x}_{n}\|^{2}
	\leq 6\,|x[1]|^{-2}\, \|\nu_{n}\|^{2}_{2} 
	\leq 6\, \mu^{-2}\, \|\nu_{n}\|^{2}_{2} 
\end{equation*}
for sufficiently small $\|\nu_{n}\|$ and using that $\|x_{n}\|^{2} \geq |x[1]|^{2} \geq \mu^{2} > 0$ for all $n$.
Inserting these inequalities into \eqref{equ:proofErrBound1} and taking the expectation, one obtains the first lines of \eqref{equ:ErrBound2}.
Similarly, if $\|\nu_{n}\|$ is sufficiently large, Lemma~\ref{lem:ErrorMatrix} and \ref{lem:EstErrEigenv} yield
$\|x_{n} - \widetilde{x}_{n}\|^{2} \leq 7\sqrt{3/2} \|\nu_{n}\|_{2}$.
Inserting this inequality into \eqref{equ:proofErrBound1} and taking the expectation one obtains the second lines of \eqref{equ:ErrBound2}.
\end{proof}

\begin{remark}
\label{rem:LowHighSNR}
As in Lemma~\ref{lem:EstErrEigenv}, the first line of \eqref{equ:ErrBound2} characterizes the high SNR regime and the second line the low SNR regime.
Sufficiently small/large  $\|\nu\|$ means here, that the conditions of Lemma~\ref{lem:EstErrEigenv} have to be satisfied for every $x_n$.
So for high SNR, we need that
$\|x_{n}\| \geq (3/\sqrt{2})\, \|\Delta Q_{n}\| \geq (3/\sqrt{2})\, \|\nu_{n}\|$ for all $n=1,\dots,N-1$. If this is satisfied, we can square both side and sum over all $n$.
This yields $\|x\|^{2} + (N-2) |x[1]|^{2} \geq \tfrac{9}{2}\, \|\nu\|^{2}$.
\end{remark}

\begin{theorem}
\label{thm:ERRBound_Phi}
Let $\Phi = \{\phi_{m,n}\}$ be the measurement ensemble defined in \eqref{equ:MeasVect}.
For any $\mu > 0$ let 
\begin{equation*}
	x \in \mathcal{S}_{\Phi}(\mu) := \{x \in \mathcal{S}_{\Phi} : |x[n]| > \mu\ \text{for all}\ n=2,3,\dots,N-1 \}
\end{equation*}
be arbitrary and let $\widetilde{x}$ be the reconstructed vector using the algorithm of Sec.~\ref{sec:Reconstruction}.\\
If the measurement noise $\nu_{m,n}$ in \eqref{equ:NoisyMeas} are i.i.d. random variables then there exists a constant $C_{3}(N) = C_{3}(N,\mu)$ such that
\begin{equation}
\label{equ:ErrBoundPhi}
	E\left[ \|x - \widetilde{x}\|^{2} \right]
	\leq C_{3}(N)\, E[ \|\nu\|^{2} ]  \qquad  \text{for sufficiently small}\ E[ \|\nu\|^{2} ]
\end{equation}
and where $C_{3}(N)$ satisfies
\begin{equation*}
	C_{3}(N) = 
	\left\{
	\begin{array}{lll}
	\mathcal{O}(1) & \text{if} & \gamma < 1/2\\
	\mathcal{O}(N) & \text{if} & \gamma = 1/2\\
	\mathcal{O}([2\gamma]^{N}) & \text{if} & \gamma > 1/2
	\end{array}\right.
	\qquad\text{and where}\qquad
	\gamma \leq \frac{\|x\|^{2}_{\infty}}{\mu^{2}}\;.
\end{equation*}
\end{theorem}

\begin{proof}
Since $x_{n} = (x[n],x[n+1])^{\T}$, the overall estimation error can be expressed as the sum of the errors in every step $n$, i.e.
\begin{multline*}
	\|x - \widetilde{x}\|^{2}
	= \frac{1}{2}\left( \sum^{N-1}_{n=1} \|x_{n}\E^{\I\Delta\theta_{n}} - \widetilde{x}_{n}\|^{2} + \left| x[1] - \widetilde{x}[1] \right|^{2} + \left| x[N] - \widetilde{x}[N] \right|^{2} \right)\\
	\leq \sum^{N-1}_{n=1} \|x_{n}\E^{\I\Delta\theta_{n}} - \widetilde{x}_{n}\|^{2}
	\leq 2 \sum^{N-1}_{n=1} \|x_{n} - \widetilde{x}_{n}\|^{2} + 2\gamma \sum^{N-1}_{n=2} \|x_{n-1}\E^{\I\Delta\theta_{n-1}} - \widetilde{x}_{n-1}\|^{2}
\end{multline*}
where we used \eqref{equ:errorStepN} to obtain the last line.
Inserting iteratively \eqref{equ:errorStepN} into the last sum, we end up with
\begin{equation*}
	\|x - \widetilde{x}\|^{2}
	\leq 2 \sum^{N-2}_{m=0} (2\gamma)^{m} \sum^{N-m-1}_{n=1} \|x_{n} - \widetilde{x}_{n}\|^{2}
	\leq 12 \sum^{N-2}_{m=0} (2\gamma)^{m} \sum^{N-m-1}_{n=1} \frac{\|\nu_{n}\|^{2}}{\|x_{n}\|^{2}}
\end{equation*}
where for the last inequality we insert the bounds of Lemma~\ref{lem:ErrorMatrix} and \ref{lem:EstErrEigenv}.
Now we take the expectation on both sides and use that all entries of $\nu$ are i.i.d. random variables, such that $E[\|\nu_{n}\|^{2}]$ is actually a constants, independent of $n$.
Moreover, we know that $\|x_{n}\|^{2} \geq 2\, \mu^{2}$ such that
\begin{equation*}
	E[ \|x - \widetilde{x}\|^{2} ]
	\leq \frac{6}{\mu^{2}}\, E\left[\| \nu_{n}\|^{2} \right] \sum^{N-2}_{m=0} [N-m-1]\, (2\gamma)^{m}
	= C_{3}(N)\, E[\|\nu\|^{2}]
\end{equation*}
because $E[\|\nu\|^{2}] = (N-1)\, E[\|\nu_{n}\|^{2}]$ and with the constant 
\begin{equation*}
	C_{3}(N)
	= \frac{6\, \mu^{-2}}{N-1} \sum^{N-2}_{m=0} [N-m-1]\, (2\gamma)^{m}
	= \frac{6\, \mu^{-2}}{(2\gamma -1)^2}\, \frac{(2\gamma)^{N} - 2\gamma N + N - 1}{N-1}\;.
\end{equation*}
If $\gamma=1/2$ then $C_{3}(N) = 3\, \mu^{-2}\,N$.
\end{proof}

The derived bounds show that our reconstruction scheme provides stable phase retrieval in the presence of noise in the subsets $\mathcal{S}_{\Psi}(\mu)$ and $\mathcal{S}_{\Phi}(\mu)$.
The stability behavior shows a very similar behavior as the one for SDP-based signal recovery \cites{Candes_PhaseLift,CandesLi_FCM13}, even though our recovery scheme is of completely algebraic nature.
In particular, at low noise power, the squared error is proportional to the noise power $E[\|\nu\|^{2}]$. This is the same behavior as for SDP based recovery schemes based on random measurements vectors \cites{Candes_PhaseLift,CandesLi_FCM13}.
However in \cite{CandesLi_FCM13} it was shown that for SDP methods with random measurements the constant $C_1$ in \eqref{equ:ErrBound2} decreases proportional with $1/N$, i.e. the performance improves at higher dimensions. In our algorithm, the constant is independent of the dimension $N$.

For the measurement ensemble $\Phi$, Theorem~\ref{thm:ERRBound_Phi} shows that the performance degrades with increasing dimension $N$ due to the error propagation in the recovery algorithm.
This is expressed by the dependency of the constant $C_{3}(N)$ on $N$. In general, $C_{3}(N)$ increases monotonic with $N$.
However, if the parameter $\gamma$ (which describes the degree of the error propagation) is smaller than $1/2$, then there exists an upper bound $C_{0}$ such that $C_{3}(N) \leq C_{0}$ for all $N$.
If, on the other hand, $\gamma > 1/2$ the constant $C_{3}(N)$ grows exponentially with $N$.

\section{Signal Recovery via SDP}
\label{sec:SDP}

Section~\ref{sec:Reconstruction} provides a fast and efficient recovery algorithm for the measurement vectors $\Phi$ and $\Psi$. 
Nevertheless, one might expect that optimization techniques for phase retrieval, as promoted in \cites{CandesEldar_PhaseRetrieval,Candes_PhaseLift}, might be more robust against measurement errors.
Therefore it seems to be desirable to apply optimization techniques also for the measurement ensembles $\Phi$ and $\Psi$.
We are going to show in this section, that if the signal is measured with $\Phi$ or $\Psi$, then signal recovery is also possible by a semidefinite program (SDP).

One easily sees, that the phase retrieval problem \eqref{equ:PRProblem} can be reformulated as a rank minimization problem 
\begin{equation}
\label{equ:RankMin}
	\begin{array}{ll}
		\text{minimize} &  \rank(X)\\
		\text{subject to} & \mathcal{A}_{V}(X) = b\;,\quad X \succeq 0\;.		
	\end{array}
\end{equation}
Indeed, assuming that only one rank $1$ solution exists, which is the original signal. Then it is clear that the rank minimization \eqref{equ:RankMin} yields the same solution as \eqref{equ:PRProblem}, and $x$ can be recovered, up to a unitary factor, by factorizing the solution $X$.
However, solving \eqref{equ:RankMin} is an NP hard problem. Therefore, the following convex relaxation, known as PhaseLift \cites{CandesEldar_PhaseRetrieval,Candes_PhaseLift}, has been proposed:
\begin{equation}
\label{equ:SDP}
	\begin{array}{ll}
		\text{minimize} & \trace(X)\\
		\text{subject to} & \trace(V_l^{*}\, X) = b[l]\;,\quad l=1,\dots,L\\
		& X \succeq 0\;.
	\end{array}
\end{equation}
This is a standard SDP for which a variety of efficient solvers have been developed in the recent years. 
In general the two programs \eqref{equ:RankMin} and \eqref{equ:SDP} are not equivalent. 
However if the measurement mapping $\Op{A}_{V}$ satisfies the conditions of the following lemma, then both programs have the same solution \cites{CandTao_IT10,Candes_PhaseLift}.
In fact, it was even noticed in \cites{Demanet_PhaselessLinMeas13,CandesLi_FCM13} that if $\mathcal{A}_{V}$ satisfies these conditions then the feasible set of \eqref{equ:SDP} reduces to the single point $X = x x^{*}$. So actually trace minimization in \eqref{equ:SDP} is unnecessary.

\begin{lemma}
\label{lem:SufCondSDP}
If for a given vector $x \in \CN^{N}$ the measurement mapping $\mathcal{A}_{V}$ satisfies the following two conditions
\begin{enumerate}
\item There exists an $Y$ in the range of $\mathcal{A}^{*}_{V}$ such that $Y_{\mathcal{T}_{x}} = 0$ and $Y_{\mathcal{T}^{\bot}_{x}} \succ 0$.
\item The restriction of $\mathcal{A}_{V} : \mathcal{H}_{N} \to \RN^{L}$ to $\mathcal{T}_{x}$ is injective.
\end{enumerate}
then $X = x x^{*}$ is the only matrix in the feasible set of \eqref{equ:SDP}, i.e. $X$ is the unique solution of \eqref{equ:SDP}.
\end{lemma}

The matrix $Y$ is often called a \emph{dual certificate}.
For the sake of completeness we provide a short proof of this lemma which may similarly be found in \cites{Candes_CDP13,Demanet_PhaselessLinMeas13}.

\begin{proof}
Let $X^{'} = X + H$ be a matrix in the feasible set of \eqref{equ:SDP}.
The goal is to show that $H = 0$. By assumption $H \in \H_{N}$ and $H \in \NS(\mathcal{A}_{V})$ and we can write $H = H_{\mathcal{T}} + H_{\mathcal{T}^{\bot}}$. Since $X^{'}\succeq 0$, it follows for all $y \in \CN^{N}$ with $\left\langle y,x\right\rangle = 0$ that
\begin{equation*}
	y^{*} X^{'} y
	= y^{*}\left( x x^{*} + H_{\mathcal{T}} + H_{\mathcal{T}^{\bot}}\right) y
	= y^{*}\,H_{\mathcal{T}^{\bot}}\, y \geq 0\;.
\end{equation*}
Because the range spaces of $H_{\mathcal{T}^{\bot}}$ and of $H^{*}_{\mathcal{T}^{\bot}}$ are contained in orthogonal complement of $\cls\{x\}$ this shows that $H_{\mathcal{T}^{\bot}} \succeq 0$.
Since $Y \in \RS(\mathcal{A}^{*}_{V}) = \NS(\mathcal{A}_{V})^{\bot}$, we have $\left\langle H,Y \right\rangle = 0$ and because $Y_{\mathcal{T}}=0$, it follows that $\left\langle H,Y \right\rangle = \left\langle H_{\mathcal{T}^{\bot}},Y_{\mathcal{T}^{\bot}} \right\rangle = 0$.
But since $Y_{\mathcal{T}^{\bot}} \succ 0$, this shows that $H_{\mathcal{T}^{\bot}} = 0$.
By injectivity of $\mathcal{A}_{V}$ on $\mathcal{T}$ also $H_{\mathcal{T}} = 0$ such that $H = 0$ and therefore $X^{'} = X$.
\end{proof}

Next we are going to show that the measurement mappings associated with the vectors
$\Phi = \{ \phi_{m,n} \}$ and $\Psi = \{ \phi_{m,n} \}$, as defined in \eqref{equ:MeasVect}, satisfy the sufficient conditions of Lemma~\ref{lem:SufCondSDP}.
These properties will easily follow from the particular construction of the vectors $\phi_{m,n}$ and $\psi_{m,n}$ based on a $2/4$-tight uniform frame.
Therefore, we restate Theorem~\ref{thm:Balan} with the particular frame $\{a_{m}\}$ given in \eqref{equ:alpha_2D}.
For this particular case, it states that
\begin{equation}
\label{equ:Balan2x2}
	Q = \frac{3}{2} \sum^{4}_{m=1} \big\langle Q , A_{m} \big\rangle \left[A_{m} - \tfrac{1}{3} I_{2}\right]
	\quad\text{for all}\quad Q \in \mathcal{P}_{1}(\CN^{2})
\end{equation}
with $A_{m} = a_{m} a^{*}_{m}$ and where $\mathcal{P}_{1}(\CN^{2})$ stands for the set of all self-adjoint rank-one projections on $\CN^{2}$, i.e. the set of all Hermitian rank-one matrices of the form $Q = x x^{*}$.
Since $\left\langle Q , A_{m}\right\rangle$ is the Hilbert-Schmidt inner product of $Q$ and $A_{m}$, the above relation can be interpreted in the sense that $\{A_{m}\}^{4}_{m=1}$ forms a frame for $\mathcal{P}_{1}(\CN^{2})$ with dual frame $\{ \widetilde{A}_{m}:=A_{m} - \tfrac{1}{3}\, I_{2} \}^{4}_{m=1}$. Consequently, one also has
\begin{equation}
\label{equ:BalanDual}
	Q = \frac{3}{2} \sum^{4}_{m=1} \big\langle  Q , \widetilde{A}_{m} \big\rangle\, A_{m}
	\quad\text{for all}\quad Q \in \mathcal{P}_{1}(\CN^{2})\;.
\end{equation}
We refer to \cite{Balan_RecWithoutPhase_06} for more details and for a proof of this statement.
With these preparations, we are able to show that our particular measurement vectors satisfy the conditions of Lemma~\ref{lem:SufCondSDP}.
First we prove the existence of the specific dual certificate $Y$ and later establish injectivity on $\mathcal{T}_x$. 

\begin{theorem}
Let $\Phi = \{ \phi_{m,n} \}$ and $\Psi = \{ \psi_{m,n} \}$ be the set of measurement vectors as defined in \eqref{equ:MeasVect} and let $\mathcal{A}_{\Phi}$ and $\mathcal{A}_{\Psi}$ be the associated measurement mappings. Then for every $x \in \CN^{N}$ there exists a $Y \in \RS(\mathcal{A}^{*}_{\Phi})$ which satisfies
\begin{equation*}
	Y_{\mathcal{T}_{x}} = 0
	\qquad\text{and}\qquad
	Y_{\mathcal{T}^{\bot}_{x}} \succ 0
\end{equation*}
and the same holds for the set $\Psi$.
\end{theorem}

\begin{proof}
We begin with the proof for $\Phi$. Any $Y \in \RS(\mathcal{A}^{*}_{\Phi})$ is self-adjoint and has the form
\begin{eqnarray*}
	& Y = \sum^{N-1}_{n=1} \sum^{4}_{m=1} \gamma_{m,n}\, \Phi_{m,n}
	= \sum^{N-1}_{n=1} B_{n}
\end{eqnarray*}
with certain coefficients $\gamma_{m,n} \in \CN$ and with the matrices $\Phi_{m,n} = \phi_{m,n}\phi^{*}_{m,n}$ and where $B_{n} := \sum^{4}_{m=1} \gamma_{m,n}\, \Phi_{m,n}$ for every $n=1,2,\dots,N-1$.
The property $Y_{\mathcal{T}} = 0$ means that $\left\langle X,Y \right\rangle = 0$ for all $X \in \mathcal{T}_{x}$ which is satisfied if and only if $Y x = 0$ and $x^{*} Y = 0$.

\begin{figure}[t]
\begin{center}
\begin{tikzpicture}
	\draw (0,0) node {\footnotesize $0$};
	\draw (0.5,0) node {\footnotesize $\cdots$};
	\draw (3.0,0) node {\footnotesize $\cdots$};
	\draw (1,0) node {\footnotesize $0$};
	\draw (3.5,0) node {\footnotesize $0$};
	\draw (0,-0.4) node {\footnotesize $\vdots$};
	\draw (0.5,-0.4) node {\footnotesize $\ddots$};
	\draw (3.5,-0.4) node {\footnotesize $\vdots$};
	\draw (0,-1) node {\footnotesize $0$};
	\draw (1,-1) node {\footnotesize $0$};
	\draw[fill = red!15!white, draw = black!15!white] (1.3,-1.3) rectangle (2.2,-2.2);
	\draw (1.5,-1.5) node {\footnotesize $*$};
	\draw (2.0,-1.5) node {\footnotesize $*$};
	\draw (1.5,-2.0) node {\footnotesize $*$};
	\draw (2.0,-2.0) node {\footnotesize $*$};
	\draw (2.6,-1.0) node {$A_{m}$};
	\draw[<-] (2.0,-1.2) arc (180:90:0.25);
	\draw (2.5,-2.5) node {\footnotesize $0$};
	\draw (3.5,-2.5) node {\footnotesize $0$};	
	\draw (0.0,-2.9) node {\footnotesize $\vdots$};
	\draw (3.0,-2.9) node {\footnotesize $\ddots$};
	\draw (3.5,-2.9) node {\footnotesize $\vdots$};
	\draw (0,-3.5) node {\footnotesize $0$};
	\draw (0.5,-3.5) node {\footnotesize $\cdots$};
	\draw (2.5,-3.5) node {\footnotesize $0$};
	\draw (3.0,-3.5) node {\footnotesize $\cdots$};
	\draw (3.5,-3.5) node {\footnotesize $0$};
	\draw (-1.5,-1.75) node {$\Phi_{m,n} = $};
	\draw[rounded corners = 2mm] (-0.25, 0.25) -- (-0.5,-0.2) -- (-0.5,-3.3) -- (-0.25,-3.75);
	\draw[rounded corners = 2mm] (3.75, 0.25) -- (4.0,-0.2) -- (4.0,-3.3) -- (3.75,-3.75);
	\draw[->] (1.8,0.5) -- (1.5,0.5) -- (1.5,0.27);
	\draw (2.8,0.5) node {\footnotesize $n$-th column};
	%
	\draw[fill = red!15!white, draw = black!15!white] (6.8, 0.20) rectangle (9.20,-2.2);
	\draw (7.0,0.0) node {\footnotesize $*$};
	\draw (9.0,0.0) node {\footnotesize $*$};
	\draw (7.5,0.0) node {\footnotesize $0$};
	\draw (8.0,0) node {\footnotesize $\cdots$};
	\draw (8.5,0.0) node {\footnotesize $0$};
	\draw (7.0,-0.5) node {\footnotesize $0$};
	\draw (7.0,-0.9) node {\footnotesize $\vdots$};
	\draw (7.0,-1.5) node {\footnotesize $0$};
	\draw (7.0,-2.0) node {\footnotesize $*$};
	\draw (7.5,-2.0) node {\footnotesize $0$};
	\draw (8.0,-2.0) node {\footnotesize $\cdots$};
	\draw (8.5,-2.0) node {\footnotesize $0$};
	\draw (9.0,-2.0) node {\footnotesize $*$};
	\draw (9.0,-0.5) node {\footnotesize $0$};
	\draw (9.0,-0.9) node {\footnotesize $\vdots$};
	\draw (9.0,-1.5) node {\footnotesize $0$};
	\draw (7.5,-0.4) node {\footnotesize $\ddots$};
	\draw (8.5,-1.4) node {\footnotesize $\ddots$};
	\draw (8,-1) node {\footnotesize $0$};
	\draw ( 9.5,0) node {\footnotesize $0$};
	\draw (10.0,0) node {\footnotesize $\cdots$};
	\draw (10.5,0) node {\footnotesize $0$};
	\draw (10.5,-0.4) node {\footnotesize $\vdots$};
	\draw (10.5,-1.9) node {\footnotesize $\vdots$};
	\draw ( 7.0,-2.5) node {\footnotesize $0$};
	\draw ( 9.5,-2.5) node {\footnotesize $0$};
	\draw (10.5,-2.5) node {\footnotesize $0$};	
	\draw ( 7.0,-2.9) node {\footnotesize $\vdots$};
	\draw (10.0,-2.9) node {\footnotesize $\ddots$};
	\draw (10.5,-2.9) node {\footnotesize $\vdots$};
	\draw (7.0,-3.5) node {\footnotesize $0$};
	\draw (7.5,-3.5) node {\footnotesize $\cdots$};
	\draw (9.0,-3.5) node {\footnotesize $\cdots$};
	\draw (9.5,-3.5) node {\footnotesize $0$};
	\draw (10.0,-3.5) node {\footnotesize $\cdots$};
	\draw (10.5,-3.5) node {\footnotesize $0$};
	\draw (5.5,-1.75) node {$\Psi_{m,n} = $};
	\draw[rounded corners = 2mm] (6.75, 0.25) -- (6.5,-0.2) -- (6.5,-3.3) -- (6.75,-3.75);
	\draw[rounded corners = 2mm] (10.75, 0.25) -- (11.0,-0.2) -- (11.0,-3.3) -- (10.75,-3.75);
	\draw[->] (8.5,0.5) -- (9.0,0.5) -- (9.0,0.27);
	\draw (7.1,0.5) node {\footnotesize $(n+1)$-th column};	
\end{tikzpicture}
\end{center}
\caption{The structure of the matrices $\Phi_{m,n}$ and $\Psi_{m,n}$. The $4$ non-zero entries in each matrix are symbolized by ``$*$''.}
\label{fig:Matrix}
\end{figure}

For fixed $n$ consider the matrices $\Phi_{m,n}$. By the definition of $\phi_{m,n}$ all entries of $\Phi_{m,n}$ are zero apart from the entries at position $(n,n)$, $(n,n+1)$, $(n+1,n)$,  and $(n+1,n+1)$. So $\Phi_{m,n}$ is zero apart from a $2\times 2$ square block on the diagonal at position $n$ (cf. Fig.~\ref{fig:Matrix}).
This $2 \times 2$ diagonal block is equal to $A_{m} = a_{m} a^{*}_{m}$ with vectors $a_{m}$ defined in \eqref{equ:alpha_2D}.
We have to find $\{\gamma_{m,n}\}$ such that
\begin{eqnarray*}
	& x^{*} Y
	= \sum^{N-1}_{n=1} x^{*} B_{n} = 0
	\qquad\text{and}\qquad
	Y x
	= \sum^{N-1}_{n=1} B_{n} x = 0	
\end{eqnarray*}
which is satisfied if $x^{*} B_{n} = 0$ and $B_{n} x = 0$ for all $n=1,\dots,N-1$.
Because of the special structure of the matrices $\Phi_{m,n}$, we have
\begin{eqnarray}
\label{equ:proofY1}
	& x^{*} B_{n}
	= x^{*} \sum^{4}_{m=1} \gamma_{m,n} \Phi_{m,n}
	= x^{*}_{n}\, \sum^{4}_{m=1} \gamma_{m,n}\, A_{m}
\end{eqnarray}
where we defined  $x_{n} := (x[n],x[n+1])^{\T}$.
For every $n$ we can always find a $q_{n} \in \CN^{2}$, with $\|q_{n}\| = 1$ such that $x^{*}_{n}\, q_{n} = q^{*}_{n} x_{n} = 0$,
and we know from \eqref{equ:BalanDual} that there exist coefficients $\gamma_{m,n} = \langle  \widetilde{A}_{m} , Q_{n} \rangle$ such that 
\begin{eqnarray*}
\label{equ:proofY2}
	& Q_{n} := q_{n}\, q^{*}_{n}
	= \sum^{4}_{m=1} \gamma_{m,n}\, A_{m}\;.
\end{eqnarray*}
By this construction, we have $x^{*}_{n} Q_{n} = Q_{n} x_{n} = 0$.
Together with \eqref{equ:proofY1} this shows that we found coefficients $\gamma_{m,n}$ such that $x^{*} B_{n} = B_{n} x = 0$ for all $n=1,\dots,N-1$ and consequently $Y_{\mathcal{T}_{x}} = 0$.

Moreover, all matrices $Q_{n}$ are positive definite, and by the above construction we have therefore $y^{*} Y y = \sum^{N-1}_{n=1} y^{*}_{n} Q_{n} y_{n} \geq 0$ for all $y \in \CN^{N}$ and where the sum is zero only if $y$ is a scalar multiple of $x$. Therein we defined, similar as above, $y_{n} = (y[n],y[n+1])^{\T}$. Consequently, we have $Y_{\mathcal{T}^{\bot}_{x}} \succ 0$.

The proof for $\Psi$ is basically the same. The only difference is that the matrices $\Psi_{m,n} = \psi_{m,n}\psi^{*}_{m,n}$ have the structure shown on the right hand side of Fig.~\ref{fig:Matrix}, namely all entries of $\Psi_{m,n}$ are zero apart from the entries at position $(1,1)$, $(1,n+1)$, $(n+1,1)$, and $(n+1,1)$. Therefore, equation \eqref{equ:proofY1} now reads
\begin{eqnarray*}
	& x^{*} B_{n}
	= x^{*}\sum^{4}_{m=1} \gamma_{m,n}\Psi_{m,n}
	= x^{*}_{n}\, \sum^{4}_{m=1} \gamma_{m,n}\, A_{m}
\end{eqnarray*}
with $x_{n} := (x[1],x[n+1])^{\T}$. The rest of the proof is exactly the same as above and therefore omitted.
\end{proof}

Next, we are going to prove that $\Phi$ and $\Psi$ satisfy the injectivity condition of Lemma~\ref{lem:SufCondSDP}. As a preparation we first derive the general structure of the null space of $\mathcal{A}_{\Phi}$ and $\mathcal{A}_{\Psi}$.

\begin{lemma}
\label{lem_NS}
Let $\Phi = \{ \phi_{m,n} \}$ and $\Psi = \{ \psi_{m,n} \}$
be the sets of measurement vectors as defined in \eqref{equ:MeasVect} and let $\mathcal{A}_{\Phi} : \H_{N} \to \RN^{L}$ and $\mathcal{A}_{\Psi} : \H_{N} \to \RN^{L}$ be the associated measurement maps.
Then the null spaces of $\mathcal{A}_{\Phi}$ and $\mathcal{A}_{\Psi}$ are given by
\begin{align*}
	\NS(\mathcal{A}_{\Phi}) = \{ X \in \H_{N}\ :\ &[X]_{n,n} = 0,\ n=1,\dots,N\ \text{and}\\
	  &[X]_{n,n+1}=[X]_{n+1,n} =0,\ n=1,\dots,N-1\}\\
	\NS(\mathcal{A}_{\Psi}) = \{ Y \in \H_{N}\ :\ &[Y]_{n,n} = [Y]_{1,n}=[Y]_{n,1} = 0,\ n=1,\dots,N \}\;.
\end{align*}
\end{lemma}

\begin{remark}
Thus all matrices $X \in \NS(\mathcal{A}_{\Phi})$ and $Y \in \NS(\mathcal{A}_{\Psi})$ are Hermitian and have the form
\begin{equation}
\label{equ:NS_Phi}
	X = 
	\begin{pmatrix}
	   0&0& * & & & *&*\\
	   0&0&0&*& & &*\\
	   *&0&0&0&*& & \\
	    & \ddots&\ddots&\ddots&\ddots&\ddots& \\
	    & & *& 0& 0&0&*\\
	   *& & & *& 0&0&0\\
	   *&*& & & *&0&0
	\end{pmatrix}
	\quad\text{and}\quad
	Y = 
	\begin{pmatrix}
	   0&0&0&\dots&0\\
	   0&0&*&\dots&*\\
	   0&*&0& &*\\	   
	   \vdots&\vdots& &\ddots&*\\
	   0&*&*&*&0
	\end{pmatrix}.	
\end{equation}
\end{remark}

\begin{proof}
We begin with the proof for $\Phi$.
A matrix $X \in \H_{N}$ belongs to $\NS(\mathcal{A}_{\Phi})$ if and only if
\begin{equation}
\label{equ:NSCond1}
	\left\langle X , \Phi_{m,n} \right\rangle = \trace(\Phi_{m,n}\,X) = 0
	\quad\text{for all}\quad
	\begin{array}{ll}
	   m=1,\dots,4\\
	   n=1,2,\dots,N-1
	\end{array}\;.
\end{equation}
By the particular structure of the matrices $\Phi_{m,n}$ (cf. Fig.~\ref{fig:Matrix}), we have
\begin{equation*}
	\left\langle X , \Phi_{m,n} \right\rangle
	=\trace(\Phi_{m,n}\,X)
	= \trace(A_{m}\,X_{n}) 
	= \left\langle X_{n} , A_{m} \right\rangle
\end{equation*}
where again $A_{m} = a_{m}\, a^{*}_{m}$ with $a_{m}$ as in \eqref{equ:alpha_2D} and where $X_{n} \in \H_{2}$ is defined by
\begin{equation*}
	X_{n}
	:= \begin{pmatrix}
		[X]_{n,n} & [X]_{n,n+1}\\
		[X]_{n+1,n} & [X]_{n+1,n+1}
	\end{pmatrix}\;.
\end{equation*}
It follows from \eqref{equ:NSCond1} that $X \in \NS(\mathcal{A}_{\Phi})$ if and only if for all $n=1,\dots,N-1$
\begin{equation}
\label{equ:NSCond2}
	\trace(A_{m}\,X_{n})	
	= \left\langle X_{n} , A_{m} \right\rangle= 0
	\quad\text{for all}\quad
	m=1,\dots,4\;.
\end{equation}
For fix $n$, the Hermitian $X_n$ can always be decomposed into the form $X_{n} = \lambda_{1}\, u_{1} u^{*}_{1} + \lambda_{2}\, u_{2} u^{*}_{2}$ with $\lambda_{1},\lambda_{2} \neq 0$, and it follows from \eqref{equ:NSCond2} that
\begin{equation*}
	\left\langle  u_2 u^{*}_2 , A_{m} \right\rangle = - \tfrac{\lambda_{1}}{\lambda_{2}}\, \left\langle u_1 u^{*}_1 , A_{m}\right\rangle
	\quad\text{for all}\quad
	m=1,\dots,4\;.
\end{equation*}
Since $\{A_{m}\}^{4}_{m=1}$ is a frame for the set of all self-adjoint rank~$1$ matrices, it follows that $u_2 u_2^{*} = -(\lambda_{1}/\lambda_{2}) u_{1} u_1^{*}$.
Thus, $X_{n}$ has rank~$1$ and therefore \eqref{equ:NSCond2} implies that $X_n = 0$.
So $X \in \NS(\mathcal{A}_{\Phi})$ if and only if $X_{n} = 0$ for all $n=1,\dots,N-1$ and this is equivalent to $X$ has the form \eqref{equ:NS_Phi}.

The proof of the second statement follows the same arguments. The only difference is that the matrices $X_{n}$ in the proof above are now given by
\begin{equation*}
	X_{n}
	:= \begin{pmatrix}
		[X]_{1,1} & [X]_{1,n+1}\\
		[X]_{n+1,1} & [X]_{n+1,n+1}
	\end{pmatrix}\;.
\end{equation*}
The rest of the proof is the same as before.
\end{proof}

Based on Lemma~\ref{lem_NS}, we can now show that the measurement ensembles $\Phi$ and $\Psi$ satisfy the second condition of Lemma~\ref{lem:SufCondSDP}. Note beforehand that Theorem~\ref{thm:FiniteDim} only shows injectivity of the mappings $\mathcal{A}_{\Phi}$ and $\mathcal{A}_{\Psi}$  on the subspaces \eqref{equ:Sets_S}.
Therefore it is sufficient for us to show that the second condition of Lemma~\ref{lem:SufCondSDP} is satisfied for all vectors $x \in \CN^{N}$ from these subspaces.

\begin{theorem}
Let $\Phi = \{ \phi_{m,n} \}$ and $\Psi = \{ \psi_{m,n} \}$ be the sets of vectors defined in \eqref{equ:MeasVect} with the associated measurement  mappings $\mathcal{A}_{\Phi} : \H_{N} \to \RN^{L}$ and $\mathcal{A}_{\Psi} : \H_{N} \to \RN^{L}$, respectively.
Then we have:
\begin{enumerate}
\item
The restriction of $\mathcal{A}_{\Phi}$ to $\mathcal{T}_{x}$ is injective for all $x \in \mathcal{S}_{\Phi}$.
\item
The restriction of $\mathcal{A}_{\Psi}$ to $\mathcal{T}_{x}$ is injective for all $x \in \mathcal{S}_{\Psi}$.
\end{enumerate}
\end{theorem}

\begin{proof}
By definition, every $X \in \mathcal{T}_{x}$ has the form
\begin{equation}
\label{equ:EntriesX}
	[X]_{n,m} = x[n]\, \overline{y[m]} + y[n]\, \overline{x[m]}\;.
\end{equation}
First, we prove the statement for $\Phi$.
Assume that $X \in \mathcal{T}_{x} \cap \NS(\mathcal{A}_{\Phi})$.
Since $X \in \NS(\mathcal{A}_{\Phi})$ it follows from Lemma~\ref{lem_NS} that
\begin{align}
\label{equ:NSPhi1}
	[X]_{n,n+1} &= x[n]\, \overline{y[n+1]} + y[n]\, \overline{x[n+1]} = 0 &
	\quad& \text{for all}\ n=1,\dots,N-1\\
	\label{equ:NSPhi2}
	[X]_{n,n} &= x[n]\, \overline{y[n]} + y[n]\, \overline{x[n]} = 0
	& \quad& \text{for all}\ n=1,\dots,N\;.
\end{align}
We are going to show that all entries \eqref{equ:EntriesX} of $X$ are equal to zero.
Since $x[n] \neq 0$ for all $n$, condition \eqref{equ:NSPhi1} can be rewritten as
\begin{equation}
\label{equ:construction}
	\overline{y[n+1]} = - \frac{\overline{x[n+1]}}{x[n]}\, y[n]\;.
\end{equation}
Inserting \eqref{equ:construction} recursively into itself, one obtains
\begin{align}
 \label{equ:xnymequality}
 x[n]\,\overline{y[m]} &= - \overline{x[m]}\,y[n]& \quad& \text{if } n-m \text{ is odd}\\
 \label{equ:xnymequality1}
 x[n]\, y[m] &= \phantom{-}x[m]\, y[n]& \quad& \text{if } n-m \text{ is even.}
\end{align}
Thus \eqref{equ:xnymequality} is equivalent to $[X]_{n,m} = 0$ if $n-m$ is odd.
In the case that $n-m$ is even, we insert \eqref{equ:xnymequality1} into \eqref{equ:EntriesX} to obtain
\begin{equation*}
	[X]_{n,m}
	= x[n]\, \frac{\overline{x[m]}\, \overline{y[n]}}{\overline{x[n]}} + y[n]\, \overline{x[m]}
	= \frac{\overline{x[m]}}{\overline{x[n]}} \left( x[n]\, \overline{y[n]} + y[n]\, \overline{x[n]} \right) = 0
\end{equation*}
where the last equality follows from \eqref{equ:NSPhi2}. So we showed that $X = 0$.

Next, we consider $\Psi$.
Let $X \in \mathcal{T}_{x} \cap \NS(\mathcal{A}_{\Psi})$, then Lemma~\ref{lem_NS} implies
\begin{align}
\label{equ:NSPsi1}
	[X]_{1,n} = x[1]\, \overline{y[n]} + y[1]\, \overline{x[n]} = 0\quad\text{for all}\ n=1,\dots,N
\end{align}
and we will show that all entries \eqref{equ:EntriesX} of $X$ are equal to zero.
Inserting \eqref{equ:NSPsi1} into \eqref{equ:EntriesX}, one obtains for arbitrary $n,m$
\begin{align*}
	[X]_{n,m}
	&= -x[n]\, \frac{y[1]}{x[1]}\, \overline{x[m]} - \frac{\overline{y[1]}}{\overline{x[1]}}\, x[n]\, \overline{x[m]}
	= -x[n]\, \overline{x[m]}\, \left( \frac{y[1]}{x[1]} + \frac{\overline{y[1]}}{\overline{x[1]}} \right)\\
	&= -x[n]\, \overline{x[m]}\, \Re\{ y[1]/x[1] \}\;.
\end{align*}
Moreover \eqref{equ:NSPsi1} shows also that $\Re\{ x[1]\, \overline{y[1]} \} = 0$ which implies $\Re\{ y[1]/x[1] \} = 0$. 
Consequently $[X]_{n,m} = 0$ for all $n,m$.
\end{proof}

We summarize the results of this section in the following corollary.

\begin{corollary}
\label{cor:SDP}
Let $\Phi = \{ \phi_{m,n} \}$ and $\Psi = \{ \psi_{m,n} \}$ be the measurement ensembles as defined in \eqref{equ:MeasVect} and let $\Phi_{m,n} = \phi_{m,n}\phi^{*}_{m,n}$ and $\Psi_{m,n} = \psi_{m,n}\psi^{*}_{m,n}$.\\
Then for every $x \in \mathcal{S}_{\Phi}$ and for every $y \in \mathcal{S}_{\Psi}$ the systems of equations
\begin{align*}
	\left\langle X, \Phi_{m,n} \right\rangle &= \left| \left\langle  x, \phi_{m,n} \right\rangle\right|^{2}\;,
	\qquad m=1,\dots,4;\ n=1,\dots,N-1\\
	\left\langle Y , \Psi_{m,n} \right\rangle &= \left| \left\langle  y, \psi_{m,n} \right\rangle\right|^{2}\,,
	\qquad m=1,\dots,4;\ n=1,\dots,N-1	
\end{align*}
have a unique solution in the set $\{ X \in \H_{N} : X \succeq 0\}$, namely $X = xx^{*}$ and $Y = yy^{*}$, respectively.
\end{corollary}

Corollary~\ref{cor:SDP} highlights the main idea of PhaseLift.
In the original problem one tries to recover $x \in \CN^{N}$ from $4N-4$ non-linear measurements where the measurement ensemble is chosen such that the mapping  $\mathcal{A}_{\phi} : x \mapsto \{ \left|\left\langle x,\phi_{m,n}\right\rangle\right|^{2} \}$ is injective.
Equivalently, the quadratic measurements can be written as linear measurements $\mathcal{A}_{\Phi} : X \mapsto \{ \left\langle X , \Phi_{m,n} \right\rangle \}$ on the cone of positive semidefinite matrices $\mathcal{C}_{N} = \{ X \in \H_{N} : X \succeq 0 \}$.
Corollary~\ref{cor:SDP} shows now that also $\mathcal{A}_{\Phi}$, which also comprises only $4N-4$ measurements, is injective, even though the dimension of the semidefinite cone $\mathcal{C}_{N}$ is much larger than $\CN^{N}$.

\section{Numerical Simulations}
\label{sec:Simulations}

In this section, we show some simulation results of phaseless signal recovery in the presence of noise using our specific measurement ensembles introduced in Sec.~\ref{sec:Reconstruction}.
Signal recovery is done with the algebraic algorithm of Sec.~\ref{sec:Reconstruction} as well as based on semidefinite optimization as discussed in Sec.~\ref{sec:SDP}.

We consider a setting as in Sec.~\ref{sec:Notations} with sets $V = \{v_{l}\}^{L}_{l=1}$ of measurement vectors. The noisy measurements are assumed to follow the model of Sec.~\ref{sec:Stability}
\begin{equation}
\label{equ:NoisyMeasSim}
	\widetilde{b}[l] = \left| \left\langle x, v_{l} \right\rangle\right|^{2} + \nu[l]\;,\quad
	l=1,\dots,L
\end{equation}
where $\nu[l]\in\RN$ is the noise term.
This model reflects many practical settings since noise directly affects the measurement at the intensity itself.
The noise components $\nu[l]$ are assumed to be i.i.d. normally distributed random variables with zero mean and variance $\sigma^{2}_{\nu}$.
The real and imaginary parts of the entries $x[n]$ of the signal $x \in \CN^N$ are assumed to be i.i.d. drawn from a normal distribution with zero mean and variance $\sigma^{2}_{x}/2$.
The signal-to-noise ratio (SNR) is defined\footnote{Note that the definition of the SNR is slightly different than in \cites{CandesEldar_PhaseRetrieval,Candes_PhaseLift} to allow for a fairer comparison between sets with different numbers of measurement vectors.} by $SNR = \sigma^{2}_{x}/\sigma^{2}_{\nu}$.

After recovering the signal $x\in\CN^{N}$ from the noisy measurements $\{\widetilde{b}[l]\}^{L}_{l=1}$ with a specific algorithm, we determine the squared error by $\min_{c\in\uC} \|x - c\, \widetilde{x}\|^{2}$ where the minimization over $c \in \uC$ accounts for the unknown phase factor, and where $\widetilde{x}$ is the recovered signal.
This experiment is repeated for many different random signal $x\in\CN^{N}$ and noise vectors $\nu\in\RN^{L}$.
Then the normalized empirical mean squared error (MSE)
\begin{equation*}
	MSE = \frac{ E[ \min_{c\in\uC} \|x - c\, \widetilde{x}\|^{2} ]}{E[ \|x\|^{2} ]}
\end{equation*}
is plotted versus the SNR.

\subsection{Algebraic signal recovery}

First, we want to show the effectiveness of the algebraic algorithm of Sec.~\ref{sec:Reconstruction} and compare the simulation results with the bounds derived in Sec.~\ref{sec:Stability}. With the signal and noise model here, the error bounds \eqref{equ:ErrBound2} for $\Psi$ become
\begin{equation}
\label{equ:ErrBoundSimPsi}
	\frac{E[\|x-\widetilde{x}\|^{2}]}{E[\|x\|^{2}]}
	\leq \left\{\begin{array}{lll}
	\frac{72}{\sigma^{2}_{x}}\, \frac{1}{SNR} & \text{if} & SNR \gtrapprox 9.5\ \text{dB} \\[1ex]
	\sqrt{\frac{3}{2}}\, \frac{42}{\sigma_{x}} \frac{1}{\sqrt{SNR}} & \text{if} & SNR \lessapprox 9.5\ \text{dB} 
	\end{array}\right.
\end{equation}
and the bound \eqref{equ:ErrBoundPhi} for $\Phi$ becomes
\begin{equation}
\label{equ:ErrBoundSimPhi}
	\frac{E[\|x-\widetilde{x}\|^{2}]}{E[\|x\|^{2}]}
	\leq \frac{12\, N}{\sigma^{2}_{x}}\, \frac{1}{SNR}
\end{equation}
where we set\footnote{Note again that this choice only influences the constants. We do not claim that this choice is the best possible. At least, it seems to be a reasonable choice which will allow a good comparison with the simulations.} $\mu^{2} = \sigma^{2}_{x}$ and $\gamma = 1/2$ in both cases.

\begin{figure}[t]
\centering
\includegraphics[width=0.8\textwidth]{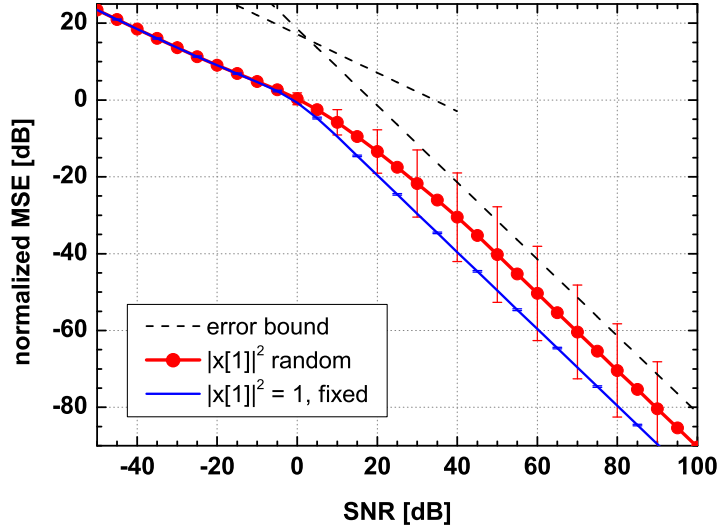}
\caption{Normalized MSE for the ensemble $\Psi$ using the algebraic recovery algorithm of Sec.~\ref{sec:Reconstruction} for $N=512$. The graphs are averaged over $3*10^4$ random signals. The error bars indicate the standard deviation.}
\label{fig:GrPsi}
\end{figure}

Fig.~\ref{fig:GrPsi} shows simulation results for the measurement ensemble $\Psi$.
The graphs are shown for signals with variance $\sigma^{2}_{x} = 1$ and dimension $N=512$.
However, simulations for different $N$ coincide exactly with the shown graphs.
We see that the bounds \eqref{equ:ErrBoundSimPsi} fairly well predict the performance of the reconstruction algorithm. Note that \eqref{equ:ErrBoundSimPsi} depends on the variance $\sigma^{2}_{x}$. So increasing the average signal energy will shift all the graphs towards lower MSE.
Moreover, Fig.~\ref{fig:GrPsi} also compares the situation where we fixed the amplitude of the first signal component ($|x[1]| = 1$) with the situation where $x[1]$ is completely random.
Recall that $x[1]$ was the common point in all $2$-dimensional phase retrieval steps. Therefore, fixing the amplitude of this basis point has a profound impact.
On the one hand, we get an improvement of about $10$~dB in the mean squared error.
On the other hand, almost all fluctuations in the performance are eliminated. So having a signal point with sufficiently large (fixed) amplitude will guarantee almost ideal performance for basically all signals.

\begin{figure}[t]
\centering
\includegraphics[width=0.8\textwidth]{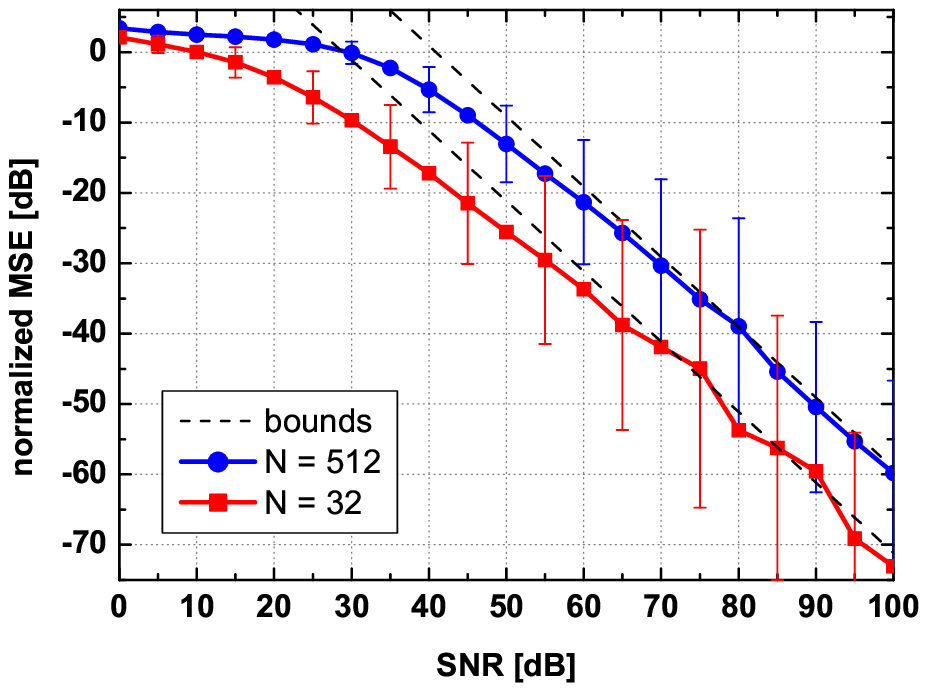}
\caption{Normalized MSE for the ensemble $\Phi$ using the algebraic recovery algorithm of Sec.~\ref{sec:Reconstruction} for $N=32$ and $N=512$. The graphs are averaged over $3*10^4$ random signals. The error bars indicate the standard deviation.}
\label{fig:GrPhi}
\end{figure}

Fig.~\ref{fig:GrPhi} shows simulations for the measurement ensemble $\Phi$ for two different dimensions $N$ and for $\sigma^{2}_{x}=1$.
Again, we see that the bound \eqref{equ:ErrBoundSimPhi} predicts fairly well the performance.
The error bound was obtained by choosing the parameter $\gamma$ in Theorem~\ref{thm:ERRBound_Phi} as $\gamma = 1/2$. If we would have chosen $\gamma > 1/2$, the performance should degrade exponentially with $N$.
However, the simulations only shows an increase of the normalized MSE linearly with $N$. This indicates that $\gamma$ is actually close to $1/2$.

\subsection{Signal recovery via SDP}

We also simulated signal recovery based on SDP as described in Sec.~\ref{sec:SDP}.
Thus, we recovered $x\in\CN^{N}$ from the noisy measurements \eqref{equ:NoisyMeasSim} by solving \eqref{equ:SDP} and incorporating the knowledge of the noise power into the side conditions:
\begin{equation}
\label{equ:SDP_noise}
	\begin{array}{ll}
		\text{mimimize} & \trace(X)\\
		\text{subject to} & \| \mathcal{A}_{V}(X) - b\|_{2} \leq L\, \sigma^{2}_{\nu}\,,\quad X \succeq 0\;.
	\end{array}
\end{equation}
Then $x$ is estimated by extracting the eigenvector $u$ associated to the largest eigenvalue $\lambda_{\mathrm{max}}$ of the solution $X$ of \eqref{equ:SDP_noise}, i.e. $\widetilde{x} = \sqrt{\lambda_{\mathrm{max}}}\, u$. 
The recovery algorithm was implemented using the common CVX solver \cite{CVX}.

\begin{figure}[t]
	\centering
     \includegraphics[width=0.8\textwidth]{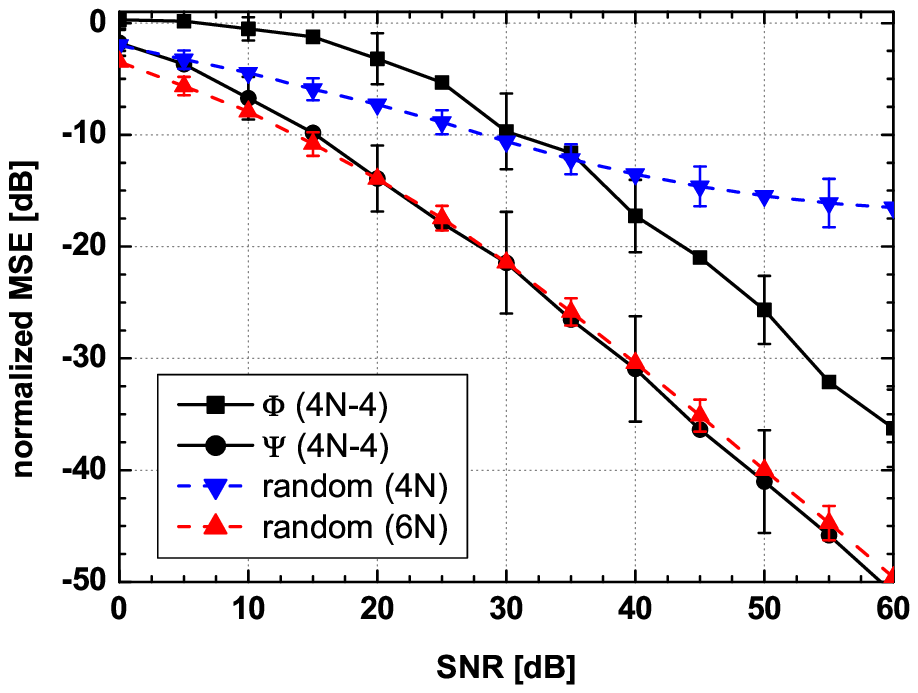}
	\caption{Signal recovery via SDP for $N=64$ and based on the sets $\Phi$ and $\Psi$ and sets with $4N$ and $6N$ random vectors.
	The graphs are averaged over 25 random signals and the error bars indicate the standard deviation.}
	\label{fig:SDP}
\end{figure}

The simulations compare the performance of the SDP algorithm for different ensembles of measurement vectors.
On the one hand, we used the sets $\Phi$ and $\Psi$ of $L= 4N-4$ measurement vectors as proposed in Sec.~\ref{sec:Reconstruction}.
For comparison we also used sets of random measurement vectors with $L = 4 N$ and $L = 6 N$ vectors. 
The random measurement vectors $V = \{ v_{l} \}^{L}_{l=1}$ are i.i.d. white noise vectors \cite{Candes_PhaseLift}, normalized to length one.
In simulations with random measurement vectors, we used a different random set $V$ for every signal.

\begin{figure}[t]
	\centering
     \includegraphics[width=0.8\textwidth]{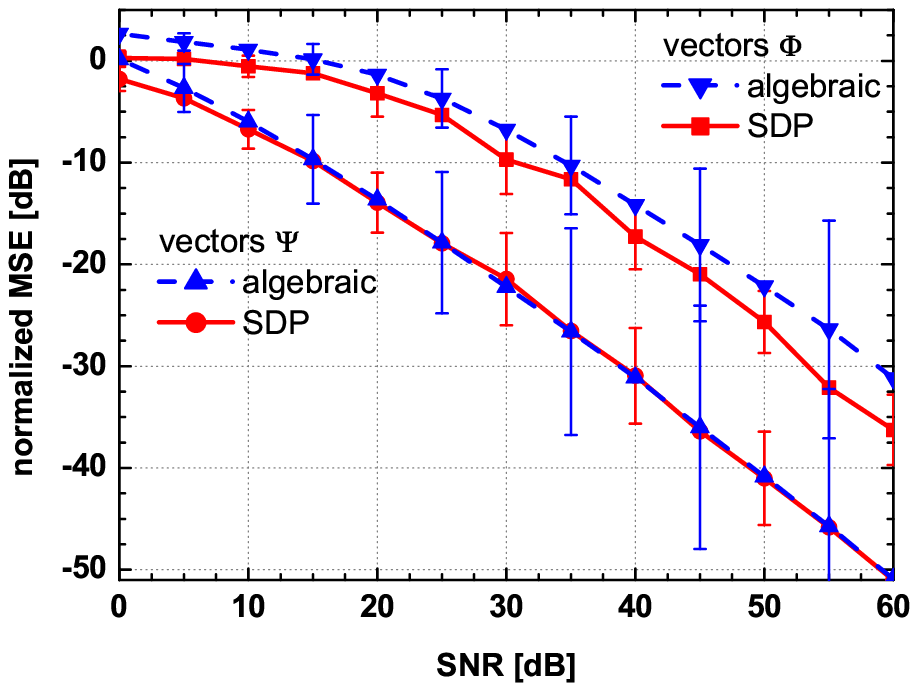}
	\caption{Signal recovery using the algebraic algorithm of Sec.~\ref{sec:Reconstruction} and SDP for the measurement vectors $\Phi$ and $\Psi$ and for dimension $N=64$.}
	\label{fig2}
\end{figure}
 
Results of the simulations are shown in Fig.~\ref{fig:SDP} for signals with $\sigma^{2}_{x} = 1$.
We observe that SDP \eqref{equ:SDP_noise} is indeed able to recover the signals from $4 N -4$ measurements taken with the measurement sets $\Phi$ and $\Psi$.
At high SNR, we see, as in the case of the algebraic recovery algorithm, that the error decreases proportional with $1/SNR$.
Also, as with the algebraic recovery, the vectors $\Phi$ perform considerable worse, compared to the set $\Psi$.
If random measurement vectors are used, signal recovery cannot be guaranteed if the number of measurements is too low.
If only $4N$ measurements are used, recovery fails completely for some signals which causes the error floor in Fig.~\ref{fig:SDP}.
With $6N$ random measurement vectors one obtains signal recovery with very high probability \cite{Candes_PhaseLift} and one obtains almost the same average performance as with the deterministic set $\Psi$.
The error bars in Fig.~\ref{fig:SDP} show that the fluctuations in the performance are generally much lower for random measurements.
However, we observed that the SDP solver using the measurement vectors $\Phi$  and $\Psi$ was significantly faster (by a factor of $10$) than using random measurement vectors, which could in part be due to the sparsity of the resulting measurement matrices.

Finally, Fig.~\ref{fig2} compares signal recovery for the sets $\Phi$ and $\Psi$ based on the algebraic algorithm of Sec.~\ref{sec:Reconstruction} with signal recovery based on the SDP \eqref{equ:SDP_noise}.
It is an interesting observation that the performance of the algebraic algorithm without optimization does similarly well in the high-SNR regime as SDP, although SDP gives a slightly better performance at low SNR.
So the performance is mainly determined by the chosen measurement vectors and not so much by the particular recovery algorithm.
Nevertheless, we notice that the standard deviation of the MSE is slightly smaller for recovery via SDP.
On the other hand, we observe that the algebraic algorithm is much faster than recovery with SDP.
In particular, the computationally complexity of the algebraic algorithm grows only linearly with the dimension $N$ of the problem, whereas the complexity of an SDP solver grows at least of the order $\mathcal{O}(N^{3})$.
This makes the algebraic algorithm very attractive for large dimensions $N$.

\section{Summary}
\label{sec:Summary}

The paper provides two sets of $4 N - 4$ measurement vectors for phase retrieval in $\CN^{N}$ together with a simple and very efficient algebraic recovery algorithm.
It was shown that these measurement vectors yield ``almost injective'' measurements, and we derived error bounds for the corresponding recovery algorithm, which show that the proposed algorithm provides stable signal recovery for any signal dimension $N$.

On the practical side, it was shown that the proposed measurement vectors can be implemented in a physical setup where the signal of interest is modulated by $4$ specific masks.
Moreover, if it is possible to guarantee a certain signal amplitude at a specific point then an almost ideal recovery performance is obtained.
It was also shown that the proposed measurements satisfy conditions which allow signal recovery via semidefinite programming (so called PhaseLift).
Numerical simulations verified the effectiveness of the proposed schemes and compared the performances.

\vspace{13pt}
\centerline{ACKNOWLEDGEMENT}
\vspace{13pt}

The authors are very thankful to the anonymous reviewers for their valuable comments.
This work was partly supported by the German Research Foundation (DFG) under grand PO~1347/2-1.


\end{document}